\documentclass{article}

\usepackage[a4paper, margin=1in]{geometry}
\usepackage{graphicx}
\usepackage[textwidth=8em,textsize=small]{todonotes}
\usepackage{natbib}
\usepackage{dsfont}

\usepackage[ruled]{algorithm2e}
\SetEndCharOfAlgoLine{}

\usepackage{amsmath}
\usepackage{amsthm}
\usepackage{amsfonts}
\usepackage{multirow}

\usepackage{amsmath,bbm,graphicx,amssymb}
\usepackage{tabularx,booktabs}
\usepackage[utf8]{inputenc}
\usepackage{url}

\usepackage{authblk}
\usepackage{setspace}
\newcommand{\argmin}{\operatornamewithlimits{argmin}}

\newtheorem{assumption}{Assumption}

\newtheorem{theorem}{Theorem}
\newtheorem{lemma}{Lemma}

\makeatother
\newcommand\independent{\protect\mathpalette{\protect\independenT{\perp}}}
\def\independenT#1#2{\mathrel{\rlap{$#1#2$}\mkern2mu{#1#2}}}

\title{Doubly regularized generalized linear models for spatial observations with high-dimensional covariates}

\author{Arjun Sondhi$^{1}$, Si Cheng$^2$, and Ali Shojaie$^2$}
\affil{$^{1}$Feinstein Institutes for Medical Research,\\	$^{2}$Department of Biostatistics, University of Washington}

\date{\today}

\begin{document}
%%%%%%%%%%%%%%%%%%%%%%%%%%%%%%

\maketitle

\begin{abstract}
A discrete spatial lattice can be cast as a network structure over which spatially-correlated outcomes are observed. A second network structure may also capture similarities among measured features, when such information is available. Incorporating the network structures when analyzing such doubly-structured data can improve predictive power, and lead to better identification of important features in the data-generating process. 
Motivated by applications in spatial disease mapping, we develop a new doubly regularized regression framework to incorporate these network structures for analyzing high-dimensional datasets.
Our estimators can be easily implemented with standard convex optimization algorithms.
In addition, we describe a procedure to obtain asymptotically valid confidence intervals and hypothesis tests for our model parameters. 
We show empirically that our framework provides improved predictive accuracy and inferential power compared to existing high-dimensional spatial methods. These advantages hold given fully accurate network information, and also with networks which are partially misspecified or uninformative. The application of the proposed method to modeling COVID-19 mortality data suggests that it can improve the prediction of deaths beyond standard spatial models, and that it selects relevant covariates more often.
\end{abstract}

%%%%%%%%%%%
\section{Introduction}
\label{p2sec:intro}
%%%%%%%%%%%

Spatial models of disease incidence and mortality are essential for understanding geographic patterns of disease spread. 
This importance was highlighted during the COVID-19 pandemic: accurate predictions of epidemic patterns were key for understanding the geographic burden of disease, allocating limited resources, and devising effective interventions. 
The left panel in Figure~\ref{fig:coviddeaths} shows the total number of COVID-19 deaths (from February 28, 2020 to September 14, 2020) in King County, Washington. 
As is commonly the case, the death counts are available over a discrete geographic space, by ZIP code. 
While broad spatial patterns are evident from the plot, it is also clear that neighboring ZIP codes do not necessarily have similar death counts, suggesting that other factors might impact COVID-19 mortality. 
In fact, making spatial estimates while ignoring ZIP code specific covariates---as in the right panel of Figure~\ref{fig:coviddeaths}---may result in over-smoothing and unsatisfactory predictions (see Section~\ref{p2sec:covid}). 
These covariates also offer insight into factors associated with the pandemic and disease burden. 
Therefore, inference for associations between covariates and the outcome is often of independent interest. 

\begin{figure}[t]
    \centering
    \includegraphics[width=7cm]{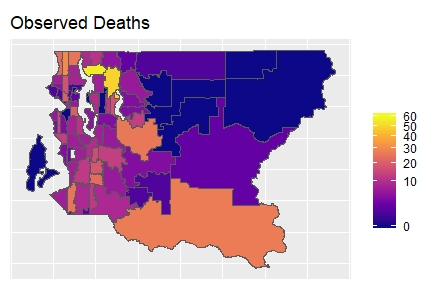}
    \includegraphics[width=7cm]{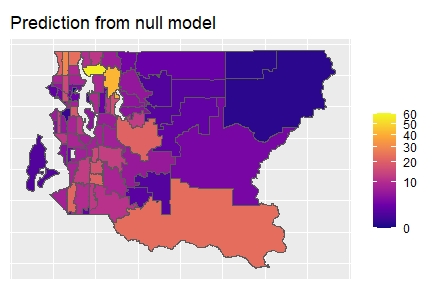}
    \caption{Left: observed COVID-19 deaths in King County, WA by ZIP code. Right: estimated COVID-19 deaths obtained by spatial smoothing, but with no covariates.}
    \label{fig:coviddeaths}
\end{figure}

Data over a discrete spatial domain can be represented as observations $y_i, i=1, \ldots, n$ on nodes of a graph $G_n = (V_n, E_n)$; here, the node set $V_n = \{1, \ldots, n\}$ corresponds to locations of observations and the edge set $E_n \subseteq V_n \times V_n$ contains undirected edges $(i, i')$ among neighboring locations, $i \sim i'$.
The conditional autoregressive (CAR) model \citep{besag1974spatial} and its intrinsic counterpart, the ICAR model \citep{besag1995conditional}, provide an elegant framework for analyzing data over discrete spatial domains. In the simple case of Gaussian observations, the ICAR model assumes that the conditional mean of each $Y_i$ is the average of its neighbors in $G_n$. Formally, denoting by $n_i$ the number of neighbors of $i$,
\begin{equation}\label{eqn:CAR}
    Y_i \mid y_{i' \ne i} \sim N\left(\sum_{i' \sim i}\frac{y_{i'}}{n_i}, \, \frac{\tau^2}{n_i}\right),
\end{equation} 
where $\tau^2$ is the shared variance parameter. 
The correlations induced by ICAR model can be captured through a spatial random effect, leading to a linear mixed model that also includes fixed effect parameters corresponding to covariates of interest, as well as an additional noise term that captures the variation not explained by the spatial component. 
The model can also be extended to non-Gaussian observations through generalized linear mixed models (GLMMs). 
However, despite recent progress \citep[e.g.,][]{guan2018computationally}, handling non-Gaussian spatial observations using GLMMs remains challenging \citep{hughes2015copcar}. 

Leveraging the conditional specification of probability distributions in CAR/ICAR, the Besag-York-Mollié model of \cite{besag1991bayesian} offers a flexible framework for spatial analysis using Bayesian hierarchical models, and is a popular choice for disease modeling. 
The prior specification of this model uses an ICAR component for spatial smoothing and an unstructured independent random effects component for location-specific noise.
Various data distributions can be accommodated in this framework by considering a latent Gaussian process over the Markov random field specified by the graph $G_n$ \citep{rue2005gaussian}. 
The introduction of the integrated nested Laplace approximation \citep[INLA;][]{rue2009approximate} has further facilitated the use of this model and its extensions \citep{leroux2000estimation, dean2001detecting}. 
The BYM2 model proposed by \citet{riebler2016intuitive} is a re-parameterization that improves interpretation of the model components, allowing for easier prior specification.

INLA offers considerable computational advantages, but the impact of its approximation on inference for fixed effect parameters is less clear. 
The recent implementation of the BYM2 model in \texttt{stan} \citep{morris2019bayesian} mitigates this issue at the cost of increased computational complexity.
Practitioners are thus faced with a tradeoff between computation and validity of inference. 
These challenges are compounded in the presence of high-dimensional covariates, which may occur when a small set of relevant covariates cannot be scientifically determined; this is illustrated in Section~\ref{sec:covid2}, when we include a larger set of covariates when analyzing COVID-19 deaths in King County. 

To overcome the above challenges, in this paper we develop a new \emph{doubly regularized} generalized linear model (GLM) framework for analyzing spatial data on discrete domains with a large number of covariates.
Our framework follows the same motivation as the CAR model \citep{besag1974spatial}, by encouraging similarity among outcomes in neighboring spatial regions in $G_n$.
However, instead of introducing correlations among neighboring regions through spatial random effects, we use an over-parametrized spatial mean surface defined by  region-specific intercepts, $\alpha_i$ for $i \in V_n$. 
More specifically, our first regularization encourages similarity among neighboring intercepts by using a \emph{network fusion penalty} that encourages $\alpha_i \approx \alpha_{i'}$ for neighboring nodes $i$ and $i'$ in $G_n$ \citep{li2019rnc}. 
Our penalty corresponds to a discrete version of the Laplace-Beltrami operator \citep{wahba1981spline} on $G_n$.
The Laplace-Beltrami operator defines thin-plate splines on manifolds \citep{duchamp2003spline}, which have a known connection to linear mixed models for Gaussian data \citep{ruppert2003semiparametric}, and are used as an alternative approach for analyzing spatial data \citep{cressie2015spatial}. 
Our approach is thus intimately related to commonly used approaches for analyzing spatial data; we further discuss the connection to mixed models in section~\ref{app:lmms} of the supplement. 
To facilitate inference for high-dimensional covariates ($p \gg n$), our approach also includes a second network regularization that encourages similarity among related covariates, as well as sparsity. 
Featuring a combination of $\ell_1$ and $\ell_2$ penalties, our doubly regularized GLM framework amounts to a convex optimization problem that can be solved efficiently for problems with large spatial domains (large $n$) and a large number of covariates (large $p$).

To infer the relevant covariates in a potentially high-dimensional model, we also develop an inference procedure for our doubly penalized generalized linear model. 
While our inference framework builds on recent developments in high-dimensional inference \citep{van2014asymptotically, javanmard2013confidence}, it is unique, as network penalty terms are generally not (semi)norms. 
This complicates the development of high-dimensional inference procedures. 
To overcome this challenge, we treat these non-norm penalty terms as part of the target loss function, which differs from the loss function minimized by the true population regression parameters. 
We control the distance between the target and true parameters, and then characterize the large sample behaviour of our proposed estimators. 
This approach allows us to incorporate partially uninformative or misspecified networks and still obtain valid inference for the association of features with the outcome. 
Although penalized regression models have recently been developed for spatially correlated data \citep{cai2019variable, chernozhukov2021}, our methods are the first, to our knowledge, to facilitate statistical inference for (non-Gaussian) generalized linear models.

We begin by describing our doubly-penalized regression methodology in Section~\ref{p2sec:funk}, including algorithms for fitting the resulting models.
We then give an overview of the theoretical details for our methods to obtain valid high-dimensional inference in Section~\ref{p2sec:asymp}.
In Section~\ref{p2sec:covid}, we analyze the King County COVID-19 death data.
The results demonstrate that our method achieves good prediction, while identifying important covariates using the proposed inference procedure. 
Our simulation studies in Section~\ref{p2sec:sims} confirm that when the networks are informative, our method provides both improved prediction accuracy and inferential power. 
They also show that our method's performances are robust to partially uninformative networks.

\section{GLMs with Feature and Unit Network Kernels}
\label{p2sec:funk}
%%%%%%%%%%%

Consider data consisting of $n$ observations $(y_1, x_1), \dots,  (y_n, x_n)$ where $y_i \in \mathbb{R}$ is the outcome and $x_i \in \mathbb{R}^p$ is the corresponding feature vector, centered and commonly scaled.
We assume the conditional mean relationship between the outcomes and features is
$$
\mathbb{E}[Y_i | X_i = x_i] = \mu(\alpha_i + x_i'\beta),
$$
where $\alpha_i \in \mathbb{R}$ are intercept terms for each unit $i$;  $\beta \in \mathbb{R}^p$ is a vector of common regression coefficients;  and $\mu$ denotes the inverse link function for the generalized linear model (GLM), which we will also refer to as the mean function.
Examples of inverse canonical link functions used for generalized linear models include $\mu(x) = x$ for Gaussian models and $\mu(x) = \exp(x)(\exp(x) + 1)^{-1}$ for binomial models.

We assume that observation units $i = 1, \dots, n$ are connected on a known (weighted) graph $G_n = (V_n, E_n)$ where $V_n = \left\{1, \dots, n \right\}$ and $E_n \subset V_n \times V_n$ is a set of undirected edges $(i, i')$ for $i \neq i'$. 
A connection between two units $i$ and $i'$ implies that they are likely to have similar outcomes $y_i$ and $y_{i'}$. 
Let $w_{ii'} \in \mathbb{R}^{+}$ denote the weight for each edge, which quantifies the strength of the similarity; in practice, we apply the binary weights $w_{ii'} = \mathds{1}\left\{(i,i') \in E_n\right\}$, but allow for arbitrary positive weights in our exposition. 
Let $A_n$ denote the weighted adjacency matrix of $G_n$, with off-diagonal elements $w_{ii'}$ and diagonal elements all zero.
%$W_n$ is a set of weights $w_{ii'} \in \mathbb{R}^{+}$ for each edge, quantifying the strength of the similarity; the classical setting of unweighted graphs corresponds to $w_{ii'} = \mathds{1}\left\{(i,i') \in E_n\right\}$. 
%Let $A_n$ be the weighted adjacency matrix of $G_n$.
Then, $L_n = D_n - A_n$ defines the graph Laplacian  \citep{chung1997spectral}, where $D_n = \text{diag}(d_1, \dots, d_n)$ and $d_i = \sum_{i' \in V_n} A_{ii'}$. 
%The Laplacian matrix will be used in the network penalties we consider.
Our doubly penalized GLM framework encourages similarity among outcomes for neighboring nodes in $G_n$ by imposing a graph Laplacian penalty on the node-specific intercepts:
$\alpha' L_n \alpha = \sum_{i \sim i'} w_{ii'} \left(\alpha_i - \alpha_{i'}\right)^2$. 
In order for the solution to be identifiable, we also add an ordinary squared $\ell_2$ penalty on $\alpha$,  $\|\alpha\|_2^2$.
In application, these penalties are respectively similar to specifying an ICAR prior for spatial smoothing and independent random effects for non-spatial variance, as in the BYM model.

Our framework also features a second penalty, $P(G_p, \beta)$, that  incorporates known similarities among features $X_j, j = 1, \dots, p$, captured by graph $G_p = (V_p, E_p)$, when such information is available. 
% We consider two possible choices for $P(G_p, \beta)$ in \eqref{eq:funk} denotes a smoothing penalty over the feature parameters $\beta$. 
We consider two possible choices of $P(G_p, \beta)$: an $\ell_2$ penalty defined based on the graph Laplacian, and an $\ell_1$ penalty, as used in graph trend filtering \citep{wang2016trend}.
A connection between two features in $G_p$ implies that they have similar associations (i.e., similar $\beta$ coefficients) with the outcome $y$. 
This type of smoothness has been  previously leveraged by \citet{li2008network} and \citet{zhao2016significance}.
Defining $A_p, D_p,$ and $L_p$ analogously as above, the $\ell_2$ fusion penalty,  
$$
P(G_p, \beta) = \frac{1}{2}\beta' L_p \beta = \frac{1}{2} \sum_{(j,j') \in E_p} w_{jj'} (\beta_j - \beta_{j'})^2,
$$
can be seen as a generalized ridge penalty that shrinks the weighted squared distance between connected features' parameters towards zero. 
% the squared difference between the coefficients will be encouraged to be small, but the distance is not shrunk exactly to zero.
Similarly, the $\ell_1$ fusion penalty can also be seen as a generalized lasso penalty \citep{tibshirani2011}. 
To this end, let $J_p$ be the incidence matrix of $G_p$, where each row of $J_p$ corresponds to an edge $(j,j')$ in $G_p$, with element $j$ of the row having value $w_{jj'}$ and element $j'$ having value $-w_{jj'}$ (in an unweighted graph, $J_p'J_p = L_p$). 
Then, the $\ell_1$ fusion penalty can be written as 
$$
P(G_p, \beta) = \| J_p \beta \|_1 = \sum_{(j,j') \in E_p} w_{jj'} |\beta_j - \beta_{j'}|, 
$$
which encourages coefficients of graph-connected features to be exactly equal.
Putting things together, the general form of our proposed estimator is given by
\begin{equation}
(\hat{\alpha}, \hat{\beta}) = \argmin_{\alpha, \beta} \left\{ \ell(y; \alpha + X\beta) + \frac{1}{2} \gamma_n \alpha' (L_n + \delta I_n) \alpha + \gamma_p P(G_p, \beta) + \lambda \| \beta \|_1 \right\},
\label{eq:funk}
\end{equation}
where $\gamma_n$, $\gamma_p$, and $\lambda$ are positive tuning parameters, and $\delta > 0$ is a  small fixed constant.
Here, $X$ denotes the design matrix of observed features $x_1, x_2, \dots, x_p$.

% In this section, we describe our framework for generalized linear models with penalization to account for both unit and feature networks.
% To achieve cohesion with respect to kernels that summarize unit and feature similarity, we propose the following estimator:
% \begin{equation}
% (\hat{\alpha}, \hat{\beta}) = \argmin_{\alpha, \beta} \left\{ \ell(y; \alpha + X\beta) + \frac{1}{2} \gamma_n \alpha' (L_n + \delta I_n) \alpha + \gamma_p P(G_p, \beta) + \lambda \| \beta \|_1 \right\},
% \label{eq:funk}
% \end{equation}
% where $\gamma_n$, $\gamma_p$, and $\lambda$ are positive tuning parameters; $\delta > 0$ can also be tuned, but for the purpose of this paper, we set it to be a small fixed value.

The optimization problem \eqref{eq:funk} has four key components: (i) the loss function $\ell$ relates the outcome $y$ to the features $X$ while allowing for a unique intercept for each observation unit; (ii) the unit network smoothing penalty smooths $\hat{\alpha}$ over $G_n$; (iii) the feature network smoothing penalty smooths $\hat{\beta}$ over $G_p$; and (iv) the standard lasso penalty enforces sparsity on the features. 
The addition of the lasso penalty allows us to obtain sparse solutions in high dimensions, even when knowledge of similarity among features does not exist (i.e., $\gamma_p = 0$). It also allows us to obtain asymptotic consistency for $\hat\beta$ at a rate that enables valid inference in high dimensions. 
We name this framework ``\texttt{g}eneralized \texttt{l}inear \texttt{m}odels with \texttt{f}eature and \texttt{u}nit \texttt{n}etwork \texttt{k}ernels", or \texttt{glm-funk}, drawing from the nomenclature of kernels as similarity matrices \citep{randolph2018kernel}.
Both of the \texttt{glm-funk} estimators described are computed by solving convex optimization problems, guaranteeing the existence of a global minimizer.

\subsection{Optimization for $\ell_2$ feature network smoothing}

To describe the algorithm for solving the \texttt{glm-funk} problem with $\ell_2$ feature network smoothing, we first rewrite the objective function in terms of 
% $\theta := \begin{pmatrix}  \alpha \\ \beta \\ \end{pmatrix}$:
$\theta := (\alpha', \beta')'$:

\begin{equation}
\hat{\theta} = \argmin_{\theta \in \mathbb{R}^{n+p}} \left\{ \ell \left(y; \tilde{X}\theta \right) + \frac{1}{2} \theta' \tilde{L} \theta + \lambda \mathcal{R}(\theta) \right\},
\label{eq:funkl2}
\end{equation}
where $\tilde{X} = [I_n \hspace{10pt} X]$, $\mathcal{R}(\theta) = \| \beta \|_1$, and
$$
\tilde{L} = \left[ \begin{array}{cc} \gamma_n (L_n + \delta I_n) & 0 \\ 0 & \gamma_p L_p \end{array} \right].
$$
Then, the optimization problem in \eqref{eq:funkl2} can be solved using a simple proximal gradient descent algorithm, given in Algorithm \ref{alg:L2}.
In our simulations and data analysis, we use a fixed step-size of $\eta^t = 0.001$, which provides reasonably fast convergence. 

\begin{algorithm}[t]
\label{alg:L2}
\SetAlgoLined
 Define $\mathcal{L}(\theta) := \ell(y; \tilde{X}\theta) + \frac{1}{2} \theta' \tilde{L} \theta$.\;
 Initialize $\theta^0$.\;
 \For{$t = 0, 1, \dots,$ until convergence of $\theta$}{
  Compute gradient, $\nabla \mathcal{L}(\theta^t) = \tilde{X}'\left( \mu(\tilde{X}\theta^t) - y \right) + \tilde{L}\theta^t$.\;
  Take a gradient step, $\tilde{\theta} = \theta^t - \eta^t \nabla \mathcal{L}(\theta^t)$.\;
  Take a proximal step, 
  \begin{align*}
  \theta^{t+1} = \text{prox}(\tilde{\theta})
               = (\tilde{\theta}_1, \dots, \tilde{\theta}_n, S_{\lambda}(\tilde{\theta}_{n+1}), \dots, S_{\lambda}(\tilde{\theta}_{n+p}))'
  \end{align*}
  where $S_\lambda(x) = \text{sign}(x) \max (0, |x| - \lambda)$.\;
 }
 \KwResult{$\hat{\theta} = \theta^{t+1}$}
 \caption{Proximal gradient descent for \texttt{glm-funk} with $\ell_2$ feature smoothing}
\end{algorithm}

\subsection{Optimization for $\ell_1$ feature network smoothing}

When using the generalized lasso penalty, the elements of $\beta$ are nonseparable in the penalty function. This leads to computational difficulties when using a non-identity GLM link function.
To overcome this challenge, we solve an alternative problem, proposed by \citet{chen2012smoothing}, by replacing the generalized lasso penalty with a smooth $\ell_\infty$ approximation of the generalized lasso penalty:
$$
f_q(\beta) = \max_{\| \nu \|_\infty < 1} \left\{ \nu' J_p \beta - \cfrac{q}{2} \| \nu \|_2^2 \right\}.
$$
Here, $q$ is a parameter that controls the approximation to the original $\ell_1$ problem; when $q = 0$, $f_q(\beta) = \| J_p \beta \|_1$. 
\citet{chen2012smoothing} prove that, for $q = \frac{\epsilon}{|E_p|}$, the absolute difference between optimal objective values of the original and approximate problems is upper bounded by $\epsilon$. 
The gradient of $f_q(\beta)$ is $J_p \nu^*$ where $\nu^* = S_\infty \left(\frac{\gamma_p J_p \beta}{q} \right)$, and $S_\infty$ is the element-wise projection operator onto the $\ell_\infty$ unit ball:
$$
S_\infty(x) = \begin{cases}
    x, &\quad \text{for } -1 \leq x \leq 1\\
    1, &\quad \text{for }  x \geq 1\\
    -1, &\quad  \text{for } x \leq -1
    \end{cases}.
$$

Replacing $\| J_p \beta \|_1$ with $f_\alpha(\beta)$, we can solve the approximate $\ell_1$ \texttt{glm-funk} problem using an accelerated proximal gradient descent algorithm \citep{beck2009fast}.
An adapted version of the algorithm presented in \citet{chen2012smoothing} is given in Algorithm \ref{alg:L1}. 
In our simulations and data analysis, we set $q = 0.001$, which provides sensible results, and a fast convergence rate.

\begin{algorithm}[t]
\label{alg:L1}
\SetAlgoLined
 Define $\mathcal{L}(\theta) := \ell(y; \tilde{X}\theta) + \frac{1}{2} \gamma_n \alpha' (L_n + \delta I_n) \alpha + \gamma_p f_q(\beta)$.\;
 Initialize $\theta^0, w^0 = \theta^0, s^0 = 1$.\;
 \For{$t = 0, 1, \dots,$ until convergence of $\theta$}{
  Compute gradient, 
  $$
  \nabla \mathcal{L}(w^t) := \tilde{X}'(\mu(\tilde{X}w^t) - y) + [\gamma_n L_n w^t  \hspace{10pt} \gamma_p J_p' \nu^*]',
  $$
  and the Lipschitz constant $C_L := \| \nabla^2 \mathcal{L}(w^t) \|_2$.\;
  Take a gradient step, $\tilde{\theta} = w^t - C_L^{-1} \nabla \mathcal{L}(\theta^t)$.\;
  Take a proximal step, 
  \begin{align*}
  \theta^{t+1} = \text{prox}(\tilde{\theta})
               = \left( \tilde{\theta}_1, \dots, \tilde{\theta}_n, S_{\lambda / C_L}(\tilde{\theta}_{n+1}), \dots, S_{\lambda / C_L}(\tilde{\theta}_{n+p}) \right)',
  \end{align*}
  where $S_{\lambda / C_L}(x) = \text{sign}(x) \max \left( 0, |x| - \frac{\lambda}{C_L} \right)$.\;
  Set $s^{t+1} = 2/(t + 3)$.\;
  Set $w^{t+1} = \theta^{t+1} + \frac{1 - s^t}{s^t} s^{t+1} (\theta^{t+1} - \theta^{t})$.\;
 }
 \KwResult{$\hat{\theta} = \theta^{t+1}$}
 \caption{Accelerated proximal gradient descent for \texttt{glm-funk} with approximate $\ell_1$ feature smoothing}
\end{algorithm}

\subsection{Prediction and tuning}

Suppose we use $n_{\mathrm{trn}}$ observations for training the \texttt{glm-funk} model, and are interested in predicting outcomes for $n_{\mathrm{tst}}$ out-of-sample observations. 
In order to make predictions on out-of-sample data, we require an estimate of the unit-level intercepts $\alpha_{\mathrm{tst}}$. 
The test sample predictions are then given as $\mu(\hat{\alpha}_{\mathrm{tst}} + X\hat{\beta})$.
Assuming we observe the entire network $G_{\mathrm{full}}$ connecting the $n_{\mathrm{trn}} + n_{\mathrm{tst}}$ units, we partition the Laplacian corresponding to $G_{\mathrm{full}}$ as:
\begin{equation*}
L_{\mathrm{full}}
= \left[ \begin{array}{cc} L_{\mathrm{trn},\mathrm{trn}} & L_{\mathrm{trn},\mathrm{tst}} \\ L_{\mathrm{tst},\mathrm{trn}} & L_{\mathrm{tst},\mathrm{tst}} \end{array} \right]
= \left[ \begin{array}{cc} L_{11} & L_{12} \\ L_{21} & L_{22} \end{array} \right]
\end{equation*}
Then, we estimate $\alpha_{\mathrm{tst}}$ as in \citet{li2019rnc}:
\begin{align*}
\hat{\alpha}_{\mathrm{tst}} 
= \argmin_{\alpha_{\mathrm{tst}}} \left\{ (\hat{\alpha}_{\mathrm{trn}}, \alpha_{\mathrm{tst}})' L_{\mathrm{full}} (\hat{\alpha}_{\mathrm{trn}}, \alpha_{\mathrm{tst}}) \right\}
= -L_{22}^{-1} L_{21} \hat{\alpha}_{\mathrm{trn}}
\end{align*}
% Note that no network knowledge for the test observations (i.e. when the training and test units are disjoint on $G_{\mathrm{full}}$) corresponds to estimating $\hat{\alpha}_{\mathrm{tst}} = 0$.
Having no network knowledge for the test observations---i.e., when the training and test units are disjoint on $G_{\mathrm{full}}$---corresponds to estimating $\hat{\alpha}_{\mathrm{tst}} = 0$.

The \texttt{glm-funk} problems involve three tuning parameters $\gamma_n$, $\gamma_p$, and $\lambda$. 
We tune these jointly using $K$-fold cross-validation to minimize the prediction error.
Ideally, the $K$ folds would be determined using non-overlapping connected components of $G_n$, but this is not always possible for arbitrary networks.
Due to the dependence among observation units over $G_n$, na\"ive cross-validation is not guaranteed to provide a good estimate of out-of-sample prediction error. 
However, as in \citet{li2019rnc}, the procedure works relatively well in practice.
Although a method for cross-validation with correlated data has been recently proposed \citep{cvcorrdata2019}, it requires knowledge of the population covariance matrix, which we do not directly assume. 
In order to efficiently determine the optimal tuning parameters, we use coordinate descent \citep{wright2015coordinate}.
Specifically, we optimize a single parameter (via $K$-fold cross-validation) while holding the others fixed, and cycle through all three parameters.
In practice, this procedure usually converges in a very small number of coordinate descent iterations.

%%%%%%%%%%%
\subsection{Related methods}
\label{p2sec:background}
%%%%%%%%%%%

As discussed in Section~\ref{p2sec:intro}, spatial regression is commonly performed with Bayesian hierarchical models using priors that induce spatial smoothing.
These models naturally require parametric assumptions, and can involve a trade-off between high computational complexity and validity of inference due to approximations used. 
In this section, we motivate our doubly regularized method by building on prior work in penalized regression over known network structures.

The recent proposal of \citet{li2019rnc} accounts for network structure among observation units by enforcing \textit{cohesion} among the units.
More specifically, it introduces the \emph{regression with network cohesion} (RNC) model, which estimates unit-level intercepts subject to a cohesion penalty over $G_n$, by solving the optimization problem,
$$
\min_{\alpha \in \mathbb{R}^n, \beta \in \mathbb{R}^p} \left\{ \ell(y; \alpha + X\beta) + \frac{1}{2} \gamma_n \alpha' (L_n + \delta I_n) \alpha \right\},
$$
where $\ell$ is a loss function (usually the negative log-likelihood), and $\gamma_n > 0$ tunes the strength of the penalty.
The similarity between observations is captured through the $n$-dimensional intercept term $\alpha$. 
The addition of $\delta I_n$ for $\delta > 0$ guarantees that a solution exists.
The penalty term, which can be written as
$$
\alpha' (L_n + \delta I_n) \alpha = \sum_{(i,i') \in E_n} w_{ii'} (\alpha_i - \alpha_{i'})^2 + \delta \sum_{i \in V_n} \alpha_i^2,
$$
implies that more strongly connected units are encouraged to have similar intercepts.
This cohesion effect implies a similarity in outcomes independent of the features $X$; connected units may still differ in their values of $x_i'\beta$.
This is similar to incorporating variance components in a generalized linear mixed model, as further discussed in section~\ref{app:lmms} of the supplement.
However, the RNC intercepts do not require distributional assumptions, and are specifically fit to optimize prediction power.
Moreover, computation is much easier, as this problem is easily solved with standard convex optimization algorithms. 
\citet{li2019rnc} show that the RNC model improves prediction for network-linked observations compared to standard methods, while maintaining the interpretability of the fixed effects in standard generalized linear models. 
However, they do not discuss high-dimensional settings and statistical inference.

A similar choice for incorporating \textit{feature} network structure is the Grace (graph-constrained estimation) penalty of \citet{li2008network}, who proposed the estimator,
$$
\hat{\beta} = \argmin_{\beta \in \mathbb{R}^p} \left\{ \ell(y; X\beta) + \frac{1}{2} \gamma_p \beta' L_p \beta + \lambda \| \beta \|_1 \right\},
$$
where, as before, $\ell$ is a loss function, and $\gamma_p, \lambda > 0$ are penalty parameters.
Unlike the RNC estimator, the Grace penalty is well-defined for high-dimensional settings, due to the regularization applied to the features.
This penalized regression encourages cohesion among $\beta$ coefficients corresponding to connected features.
The inclusion of an $\ell_1$ penalty also enforces sparsity in the solution $\hat{\beta}$. 
\citet{zhao2016significance} developed a significance test for Grace-penalized estimation, but their approach only applies to linear regression models.
This similarly holds for the method of \citet{chernozhukov2021}, which extends LASSO-based inference to space-time correlated data.

\citet{randolph2018kernel} account for two-way structured data in a \textit{kernel-penalized} linear regression model by solving
$$
\hat{\beta} = \argmin_{\beta} \left\{ \| y - X\beta \|_H^2 + \lambda \| \beta \|_{Q^{-1}}^2 \right\},
$$
where $H$ and $Q^{-1}$ are, respectively, $n \times n$ and $p \times p$ kernel matrices which summarize distances between the units and features. 
This can be thought of as fitting a generalized least squares (GLS) model subject to a generalized ridge penalty.
Ignoring the $\ell_1$ penalty, we can consider the proposal of \citet{li2008network} to fall within this framework, using the graph Laplacian as a kernel.
The RNC penalty of \citet{li2019rnc} is also similar to the first term in the kernel-penalized regression problem. 
Both methods penalize quantities that capture ``left-over" variation from the features with respect to a unit distance matrix; in kernel-penalized regression, the distances between residuals $y - X\beta$ are penalized, while RNC penalizes the intercepts $\alpha$.
However, GLS does not easily extend to non-Gaussian models.
Therefore, in this paper, we apply the idea of kernel penalization to generalized linear models by unifying RNC and Grace-style penalties in high-dimensional settings.
We also develop a statistical inference procedure that allows for valid hypothesis tests and confidence intervals for the regression coefficients, even when the networks are not fully informative. 

%%%%%%%%%%%
\section{Asymptotics and Inference}
\label{p2sec:asymp}
%%%%%%%%%%%

To assess the importance of covariates $X_j, j = 1, \dots, p$ for outcome $Y$ measured over a spatial domain, we are interested in obtaining valid inference for their corresponding regression coefficients $\beta$. 
Our estimator given in \eqref{eq:funk} is non-standard due to the use of the $n$-dimensional intercept term and penalty terms that are not semi-norms. 
In this section, we first investigate the large sample behaviour of $\hat{\beta}$ and $\hat{\alpha}$ estimated using the \texttt{glm-funk} estimator \eqref{eq:funk} with $\ell_1$ smoothing.
We defer discussion of the estimator with $\ell_2$ smoothing \eqref{eq:funkl2} to the Supplementary Material.
The consistency results are then used to obtain a valid inference procedure for the true regression parameters.

\subsection{Asymptotics}

We begin by describing the assumptions required for our theoretical results to hold.
First, we require that the outcomes $Y_i$ satisfy certain tail properties. 
\begin{assumption}[Tail behaviour]
\label{assump:tails}
One of the following holds: \\
(i) The centered observed outcomes $Y_i - \mathbb{E}[Y_i|X_i = x_i] = Y_i - \mu_i$ are uniformly sub-Gaussian, i.e.
$$
\max_{i = 1, \dots, n} K^2 \mathbb{E}\left[\exp \left(\cfrac{(Y_i - \mu_i)^2}{K^2} \right) - 1 \right] \le \sigma^2_0.
$$
for some constants $K, \sigma^2_0 > 0$. \\
(ii) The centered observed outcomes $Y_i - \mathbb{E}[Y_i|X_i = x_i] = Y_i - \mu_i$ are uniformly sub-exponential, satisfying 
$$
\max_{i = 1, \dots, n} \| Y_i - \mu_i \|_{\psi_1} = K_{\psi_1} < \infty,
$$
where
$$
\| Y \|_{\psi_1} = \inf \left\{ t > 0:  \mathbb{E} \exp(|Y|/t) \le 2 \right\}.
$$
\end{assumption}

These tail conditions cover a large variety of common generalized linear models.
Gaussian and binomial data satisfy the sub-Gaussian property, while Poisson and exponential outcomes have sub-exponential tails. 

We further require some conditions on the loss function $\ell$ and the design matrix $X$. 

\begin{assumption}[Loss function properties]
\label{assump:lip}
The following hold: \\
(i) The loss function $\ell: \Theta \times \Omega \rightarrow \mathbb{R}$ is integrable over all $(X,y) \in \Omega$ for each $\theta \in \Theta$. \\
(ii) For almost all $(X,y) \in \Omega$, the derivative $\nabla_\alpha \ell$ exists for all $\alpha$. \\
(iii) There exists an integrable function $g: \Omega \rightarrow \mathbb{R}$ such that $|\ell(y; \tilde{X}\theta)| \le g(X, y)$ for all $\theta \in \Theta$ and almost all $(X,y) \in \Omega$. \\
(iv) The conditional mean function $\mu$ and its derivative $\mu'$ are Lipschitz continuous with constants $C_\mu < \infty$ and $C_{\mu'} < \infty$. \\
(v) $\mu'$ is uniformly bounded away from zero, that is, $|\mu'(\cdot)|^{-1} \le U' < \infty$.
\end{assumption}

\begin{assumption}[Design scaling]
\label{assump:designL2}
The design matrix $X$ satisfies $|X_{ij}| \le R < \infty$ for all $i,j$, and scales as $\|X\|_2 = o_p\left( \sqrt{\frac{n}{\log p}} \right)$.
\end{assumption}

The loss function assumptions are fairly mild, and common in the high-dimensional inference literature \citep{van2014asymptotically, javanmard2013confidence, buhlmann2013statistical}. 
The design scaling assumption is equivalent to assuming that the maximum eigenvalue of $X'X$ grows at a rate slower than $n$, and can be shown to hold for various random designs  \citep[see Section 6.4 of][]{wainwright2019high}. It also implies that $\sqrt{\frac{\log p}{n}} = o_p(1)$ by the boundedness of $|X_{i,j}|$.

In order to state the remaining assumptions, we first define some quantities of interest.
For a generic function $f$, let $\mathbb{P}f := n^{-1} \sum_{i=1}^n \mathbb{E}[f(y_i, x_i)]$ and $\mathbb{P}_nf := n^{-1} \sum_{i=1}^n f(y_i, x_i)$.
Then, we rewrite our optimization problem as:
$$
\hat\theta = \argmin_{\theta} \left\{\mathbb{P}_n \mathcal{L}(\theta) + \lambda \mathcal{R}(\theta) \right\},
$$
where $\mathcal{L}(\theta) = \ell(\alpha_i + x_i'\beta) + \frac{1}{2} \gamma_n \alpha' (L_n + \delta I_n) \alpha$ and $\mathcal{R}(\theta) = \| \beta \|_1 + \frac{\gamma_p}{\lambda} \| J_p \beta \|_1$.
Our analysis involves two sets of parameters. 
The \textit{true parameter} is defined as $\theta^0 := \argmin_{\theta} \mathbb{P} \ell(\theta)$.
We are interested in inference for the true parameter; however, theoretical results are difficult to directly obtain since the penalty involves terms which are not semi-norms.
Therefore, we also work with the \textit{target parameter}, which we define as $\theta^* := \argmin_{\theta} \mathbb{P} \mathcal{L}(\theta)$, where $\mathcal{L}$ includes the non-$\ell_1$ penalty terms.

Our theoretical analysis is performed in a high-dimensional setting, where $n$ and $p(n)$ (and hence, the dimensions of $\alpha$ and $\beta$) are allowed to grow to infinity.
Therefore, $\theta^0$ and $\theta^*$ are dependent on $n$ and $p$, and the following assumptions apply to a sequence of data-generating processes indexed by $(n,p)$. 
Our theoretical results thus hold with high probability for large $(n,p)$. 
For ease of exposition, we do not include this dependence in our notations. 

We make the following assumptions for estimation of the target and true regression parameters.
\begin{assumption}[Compatibility condition]
\label{assump:comp}
Given a set $S \subset \left\{1, \dots, p\right\}$ with $|S| = s$, for all $c > 0$ and for all $\theta = (\alpha, \beta)'$ satisfying $\|\beta_{S^c}\|_1 + c \|J_p \beta\|_1 \le \| \alpha \|_1 + 3 \| \beta_S\|_1$, it holds that: 
$$
\cfrac{\| \alpha \|_1}{2} + \| \beta_S \|_1 \le \cfrac{\|\theta\| \sqrt{s}}{\phi(s)}
$$
for some norm $\| \cdot \|$ and constant $\phi(s) > 0$.
\end{assumption}

\begin{assumption}[Restricted strong convexity]
\label{assump:rsc}
For all $\theta = (\alpha, \beta)'$ satisfying
$$
\| \alpha - \alpha^* \|_1 + \| \beta - \beta^* \|_1 + \frac{\gamma_p}{\lambda} \left\| J_p \left( \beta - \beta^* \right) \right\|_1 \le M^*,
$$
with
$$
M^* = \cfrac{16s\lambda^2}{\rho \phi^2(S) c} + \cfrac{2\gamma_p \| J_p \bar{\beta} \|_1}{\rho},
$$
and $\lambda \ge 8\rho$, it holds that:
$$
\mathbb{P}\left( \ell(\theta) - \ell(\theta^*) \right) \ge \nabla \mathbb{P}\ell \left( \theta^* \right)' \left( \theta - \theta^* \right) + G\left( \| \theta - \theta^* \| \right),
$$
where $G(x) = c x^2$ for some constant $c > 0$. 
\end{assumption}

The compatibility and restricted strong convexity conditions are common in high-dimensional theory \citep{buhlmann2011statistics}.
Intuitively, restricted strong convexity at the optimum $\theta^*$ means that the loss function is curved sharply around $\theta^*$. 
Hence, when $\mathbb{P}(\ell \left( \theta^* \right) - \ell(\theta))$ is small, so is $\| \theta^* - \theta \|$.
\citet{negahban2012unified} describe how this condition is needed for nonlinear models, and proved that it holds for various common loss functions in sparse high-dimensional regimes, including the logistic regression deviance. 

Finally, we make assumptions on components of the true data-generating processes, and their relationship to the penalty parameters in the model.

\begin{assumption}[Sparsity]
\label{assump:sparsity}
$\beta^0$ is $s$-sparse, that is, $\| \beta^0 \|_0 = s$ with $s = O_p \left( \sqrt{n} / \log p \right)$.
\end{assumption}

\begin{assumption}[Penalty scaling]
\label{assump:scaling}
The following hold: \\
(i) $\lambda = O_p\left( \sqrt{\frac{\log p}{n}} \right)$, \\
(ii) $\gamma_p \| J_p \beta^0 \|_1 = o_p(\lambda)$, and \\
(iii) $\gamma_n \| (L_n + \delta I) \alpha^0 \|_2 = O_p\left( n^c \right)$ where $c \in (0, \frac{1}{2})$.
\end{assumption}

The sparsity assumption and scaling condition of $\lambda$ are standard rates in the high-dimensional inference literature \citep{van2014asymptotically, negahban2012unified}.

Part (ii) of Assumption \ref{assump:scaling} allows us to not observe fully informative feature networks.
It states that the quality of feature network smoothing is inversely proportional to the magnitude of its tuning parameter. 
That is, if the feature network is truly informative, we expect $\| J_p\beta \|_1 \rightarrow 0$ (at a rate faster than $\lambda$), and $\gamma_p$ can be larger.
In this scenario, the network structure impacts the $\beta$-penalization more than the lasso penalty. 
However, if the network smoothing does not capture the true structure of $\beta$, then $\| J_p\beta^0 \|_1$ will be far from zero, and $\gamma_p$ should tend to 0 at a rate faster than $\| J_p \beta^0 \|_1^{-1}$. 
In addition, if the feature network is uninformative, then $\gamma_p$ also needs to go to 0 faster than $\lambda$.
In this case, most of the penalization is driven by the lasso penalty, rather than the network structure encoded by $J_p$.
The empirical results in Section~\ref{p2sec:sims} suggest that our cross-validation approach achieves these properties data-adaptively.

Part (iii) of Assumption \ref{assump:scaling} similarly allows for some degree of non-informativeness in the unit network, and is also necessary to establish control of $\| \alpha^* - \alpha^0 \|$.
It establishes the trade-off needed between the unit network parameter $\gamma_n$ and the ridge penalty parameter $\delta$; if the unit network is informative and $\gamma_n$ is large, then $\delta$ should shrink to 0. 

Under these assumptions, we can prove that $\hat{\beta}$ and $\hat{\alpha}$ tend to the target parameters $\beta^*$ and $\alpha^*$ in $\ell_1$ norm. 

\begin{theorem}[Consistency]
\label{thm:consistency}
Under Assumptions \ref{assump:tails}-\ref{assump:scaling}, we have that
$$
\| \hat{\alpha} - \alpha^* \|_1 + \| \hat{\beta} - \beta^* \|_1 = O_p \left( \lambda + \cfrac{\gamma_p}{\lambda} \| J_p \beta^* \|_1 \right).
$$
\end{theorem}

The proof, given in the Supplementary Material, follows a similar argument as the proof for generalized sparse additive models in \citet{gsam2019}.
A key difference in our theory is handling the $n$-dimensional intercept term $\alpha$, and its corresponding $\ell_2$ penalty.
We also prove the result for outcome distributions which are not sub-Gaussian, such as Poisson and exponential data. 

This result allows for consistent estimation of the target parameters $\beta^*$ and $\alpha^*$.
In the Supplementary Material, we show that $\beta^*$ is the same as $\beta^0$, and the bias of $\alpha^*$ relative to $\alpha^0$ is controlled.
This allows for us to obtain valid inference for the true $\beta^0$, under our procedure which we describe in the next section.

\subsection{Inference}

We now describe a statistical inference procedure for the $\beta$ parameters in the \texttt{glm-funk} model.
We are specifically interested in testing the individual association between the outcome $y$ and feature $X_j$, conditional on other features $X_{-j}$.
This corresponds to the null hypotheses $H_{0,j}: \beta_j = 0$ for $j = 1, \dots, p$.
Classical inference theory does not directly apply in the high-dimensional setting.
Instead, high-dimensional inference approaches generally either involve post-selection procedures, potentially via sample splitting \citep{wasserman2009high, lee2016exact}, or constructing an asymptotically unbiased estimator from the optimal solution $\hat{\beta}$ \citep{dezeure2015high}. 
We focus on the latter approach, known as de-biasing.
%Note to Ali (for later): add a comment about post-selection, ref to Sen's

Several de-biasing procedures have been developed, including the low-dimensional projection estimator by \citet{zhang2014confidence}, the ridge projection estimator by \citet{buhlmann2013statistical}, and the desparsified lasso estimator by \citet{van2014asymptotically}. 
These methods generally consider regression models with an ordinary ridge or lasso penalty only.
The Grace test by \citet{zhao2016significance} specifically provides inference for linear regression with the $\ell_2$ Laplacian penalty.
However, this method does not account for unit-level networks, and does not extend to the case of generalized linear models. 

We consider the de-biased estimator of \citet{javanmard2013confidence}, which easily extends to generalized linear models.
Our de-biased estimator is defined as:
\begin{align*}
\hat{b} 
&= \hat{\beta} - n^{-1} M X'(\mu(\hat{\alpha} + X\hat{\beta}) - y),
\end{align*}
where $M$ is the inverse of $\hat{\Sigma} := \frac{1}{n} \nabla^2 \ell(\hat{\alpha} + X\hat{\beta})$.
The sandwich estimator of the variance-covariance matrix may also be used here.
In the case of Gaussian linear models, $\hat{\Sigma}$ requires a consistent estimator of $\sigma$, the noise standard deviation; the scaled lasso estimator of \citet{sun2012scaled} can be used to obtain such an estimate.
$M$ is computed by solving an optimization problem in which $\| \hat{\Sigma}M - I_p \|_\infty$ is minimized.
Specifically, for each $j = 1, \dots, p$, the $j$-th column of $M$, $m_j$, is defined as the solution to
$$
\min_{m \in \mathbb{R}^p} m'\hat{\Sigma}m \hspace{10pt} \text{subject to} \hspace{10pt} \|\hat{\Sigma}m - e_j\|_\infty \le q,
$$
where $e_j \in \mathbb{R}^p$ is the $j$-th basis vector.
The next result shows that under the assumptions described previously, we obtain an asymptotic distribution suitable for inference.
\begin{theorem}[Asymptotic Normality]
\label{thm:inference}
Under the conditions of Theorem \ref{thm:consistency}, as $n \longrightarrow \infty$,
$$
\sqrt{n} \left( \hat{b} - \beta^0 \right) \longrightarrow_d N\bigg(0, n^{-1} M \mathbb{E} \left[ \nabla \ell(\alpha^0 +  X\beta^0) \nabla \ell(\alpha^0 + X\beta^0)' \right] M \bigg).
$$
\end{theorem}
In the proof, given in the Supplementary Material, we establish that the target $\beta^*$ is the same as the true $\beta^0$, and derive an asymptotic rate for $\| \alpha^* - \alpha^0 \|_1$. 
We then apply Theorem \ref{thm:consistency} to show the result.
As a result, even though our target parameter $\theta^*$ is different from $\theta^0$, Theorem \ref{thm:inference} facilitates valid inference for the true regression parameters $\beta^0_j$.
Specifically, we use the corresponding test statistic:
$$
T_j = \cfrac{\sqrt{n} \hat{b}_j}{ [ M \hat{\Sigma} M ]_{jj}^{1/2} },
$$
and 100(1 - $\alpha$)$\%$ confidence interval:
$$
\hat{b}_j \pm q_{(1 - \frac{\alpha}{2})} n^{-1/2} [ M \hat{\Sigma} M ]_{jj}^{1/2}.
$$

\section{Analysis of King County COVID-19 Data}
\label{p2sec:covid}

In this section, we analyze novel coronavirus (COVID-19) death counts in King County, WA\footnote{\url{https://www.kingcounty.gov/depts/health/covid-19/data/daily-summary/extracts.aspx}}, obtained September 14, 2020 and containing data since the beginning of the pandemic. We incorporate geographical, demographic and socioeconomic information from King County GIS Open Data\footnote{\url{https://www.kingcounty.gov/services/gis/GISData.aspx}} along with the list of long-term care facilities being monitored for COVID-19\footnote{\url{https://www.kingcounty.gov/depts/health/covid-19/data/~/media/depts/health/communicable-diseases/documents/C19/LTCF-list.ashx}}.
We model the spatial domain as an unweighted graph whose nodes are  the ZIP codes (83 total for which COVID-19 data is available) in King County; two nodes are connected if they have overlapping borders. The outcome of interest is the total number of residents at each ZIP code who have died due to COVID-19. 

We conduct two analyses both using a Poisson generalized linear model with ZIP code population as an offset, and compare results from our models with those from the state-of-the-art BYM2 model \citep{riebler2016intuitive}, implemented with integrated nested Laplace approximation (INLA) \citep{rue2009approximate, lindgren2015bayesian}. 
In our first analysis (Section~\ref{sec:covid1}), we consider a small number of covariates without including the feature kernel penalty $P(G_p,\beta)$, i.e., we set $\gamma_p = 0$.
In our second analysis (Section~\ref{sec:covid2}), we include a larger number of covariates and demonstrate the utility of the feature kernel penalty, $P(G_p,\beta)$ in fusing the parameters for related covariates.

\subsection{COVID-19 model without feature network penalty}
\label{sec:covid1}

Covariates included in this analysis are race distribution (proportions of Black, White and Asian populations), median household income, education status (proportion of residents with college degree or above), age distribution (proportions of residents in four age groups: 18-29, 30-44, 45-59 and over 60), proportion of residents with medical insurance, number of hospitals, population density, number of King County Metro transit stops (as proxy for residential vs commercial neighborhood), and the number of nursing homes being monitored, each summarized on the ZIP code level. 

Covariates represented by proportions of different categories, such as race and age, are relative abundances. To mitigate the effect of spurious correlations in these \emph{compositional data}, we use the {additive} log-ratio transformation \citep{aitchison1982statistical}, omitting one category (e.g. age group below 18 and races other than Black, White or Asian). Figure~\ref{fig:exploratory} presents the distribution of race (untransformed), income, nursing homes and the number of deaths. The plots highlight correlations between race and income, and between the number of nursing homes and deaths.

\begin{figure}[t]
    \centering
    \includegraphics[width=7cm]{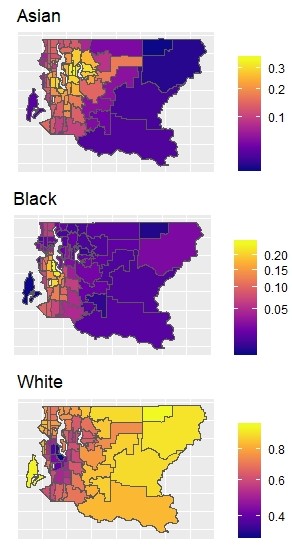}
    \includegraphics[width=7cm]{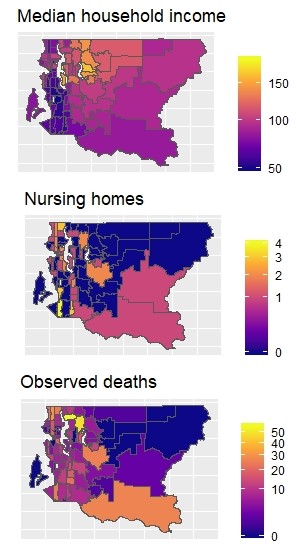}
    \caption{Distribution of race (proportions), income (median), nursing homes (count) and COVID-19 deaths by ZIP code.}
    \label{fig:exploratory}
\end{figure}

In addition to the `full' \texttt{glm-funk} model, which utilizes the spatial graph structure and the above covariates, and the `null' spatial smoothing model with no covariates visualized in Section~\ref{p2sec:intro}, we consider two benchmark models that ignore the spatial graph structure and treat all ZIP codes as independent.
One removes the graphical fusion term and assigns a common intercept to all regions, similar to a traditional Poisson GLM.
The other fits individual intercepts penalized towards zero, i.e. where $L_n=0$. 
The tuning parameters $\gamma_n$ and $\lambda$ are selected via 10-fold cross-validation, to minimize the Poisson negative log-likelihood.
We randomly assign the cross-validation folds under the constraint that no adjacent regions can be assigned to the same test set. 
This is due to the relatively small sample size and the variability in the observations, especially for regions that are far apart.
We use the estimates from the independent model with a common intercept as initial values for optimization when fitting the other two models, to obtain more stable solutions from the optimization with limited sample size. 

We also compare the proposed model with the BYM2 model \citep{riebler2016intuitive}, which is a modified version of the Besag-York-Molli{\'e} (BYM) model \citep{besag1991bayesian}. The BYM2 model specifies the linear predictor as $\log\mathbb{E}[Y_i]=\log P_i+\alpha+X_i\beta+\left[\sqrt{\rho/s}\cdot\phi_i+\sqrt{1-\rho}\cdot u_i \right]\sigma$, where $P_i$ is the offset, $\phi_i$ corresponds to the spatially correlated errors, and $u_i$ reflects non-spatial heterogeneity. For the error term, $\sigma$ is the overall standard deviation, $\rho\in[0,1]$ controls the proportion of spatial and non-spatial errors, and the scaling factor $s$ is directly determined by the graph Laplacian such that $\text{Var}(\phi_i/\sqrt{s})\approx 1$ for each $i$.
We conduct Bayesian inference for this model with the \texttt{R-INLA} package \citep{lindgren2015bayesian}. % with \texttt{rstan} (5 chains with 10,000 MCMC samples each), and adopt the same choice of priors as  \citet{morris2019bayesian}: independent Normal$(0,1)$ priors for $\alpha,\beta,\sigma,\theta$, a Beta$(0.5,0.5)$ prior for $\rho$, and the prior of $\phi$ given by $p(\phi)\propto \exp\left[-\frac{1}{2}\sum_{i\sim j}(\phi_i-\phi_j)^2\right]$. Note that under this prior for $\phi$, the maximum a posteriori (MAP) procedure leads to a similar $\ell_2$-type penalty for spatially-correlated errors as in our model. 
% \sccomment{The performance of all models is mainly assessed by the negative log-likelihood % the prediction mean-squared error (MSE) 
% with 10-fold cross-validation, and we also report the prediction root mean-squared error (RMSE) reflecting their prediction accuracy.}

Figure~\ref{fig:rslts} visualizes the residuals from each model fitted with the full data, and reports their average negative log-likelihood across the 10 cross-validation sets, with a quantile-based 95\% confidence interval.
We also report the average root mean-squared error (RMSE) to compare their prediction accuracy.
Roughly, the residuals reflect the bias of each model, while the negative log-likelihood and RMSE show the overall model prediction quality. 
The full \texttt{glm-funk} model achieves better model predictions than other models reflected by both negative log-likelihood and RMSE, since it accounts for the similarity between adjacent regions based on the spatial graph structure, trading off some region-specific bias in a data-adaptive way. 
BYM2 has a relatively high loss and prediction error due to the large variability in its predictions (ranging from 14.3 to beyond 1000 across 10 cross-validation folds), despite its smaller bias than \texttt{glm-funk}. 
For this reason, we report its median negative log-likelihood and RMSE instead of the average across cross-validation folds.
Among all frequentist models, the independence model with individual intercepts achieves lower biases in its predictions by allowing each region to have an individual baseline risk, but its overall bias-variance tradeoff is suboptimal compared to \texttt{glm-funk}.

\begin{figure}[t]
    \centering
    \includegraphics[width=16cm]{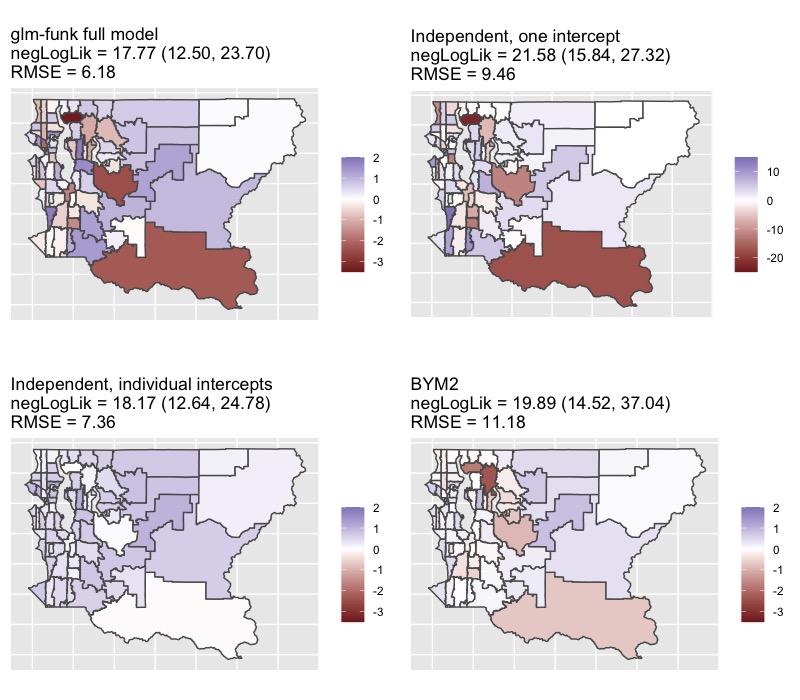}
    \caption{Residuals from each model, with cross-validated negative log-likelihood (along with 95\% CI across 10 cross-validation folds) and prediction RMSE reported in the titles. We winsorized the negative log-likelihood for BYM2, replacing the minimum (and maximum) values with the second smallest (and largest) values to mitigate the impact of outliers.}
    \label{fig:rslts}
\end{figure}

In addition to cross-validated negative log-likelihood and RMSE, we also compare the models in terms of estimated effect sizes (rate ratios). Table~\ref{tab:debiased} shows the effect sizes (posterior mean is used for the BYM2 model) along with 95$\%$ confidence intervals (for non-Bayesian models) or credible intervals (for BYM2) of the covariates in each model. 
The effect sizes are reported on the scale of standardized and transformed covariates.
The full \texttt{glm-funk} model results highlight the proportion of senior population and the number of nursing homes %race, medical insurance, transit stops, and nursing homes 
as being significant covariates. 
The independent model with one intercept identifies the majority of covariates to be statistically significant, except for the proportion of White population and the number of transit stops. 
The independent model with region-specific intercepts, in contrast, does not find any covariate effect to be significant. 
Such observation aligns with our intuition that the one-intercept independence model attributes more heterogeneity in the outcome to the covariates, while the individual-intercept model captures more variability through the spatially-varying intercepts. 
\texttt{glm-funk} achieves a balance between them by specifying region-specific, but spatially-fused intercepts.
The significance of senior population and the number of nursing homes identified by \texttt{glm-funk} matches our expectation and knowledge for COVID-19, and the correlation between the number of deaths and the number of nursing homes shown in Figure~\ref{fig:exploratory}.
%This inference is qualitatively similar to that of the independence model with region-specific intercepts; however, that model gives coefficient estimates that are more attenuated towards the null.
%This makes sense since it captures more outcome variability through the independent region-specific intercepts than the spatially-fused intercepts of \texttt{glm-funk}.
The BYM2 model does not identify any significant covariates, and tends to provide wider credible intervals compared to the frequentist methods; this reflects greater uncertainty in the estimates made by the BYM2 model, especially with small to moderate sample sizes as in our example. 

\begin{figure}[t]
    \centering
    \includegraphics[height=0.3\textheight]{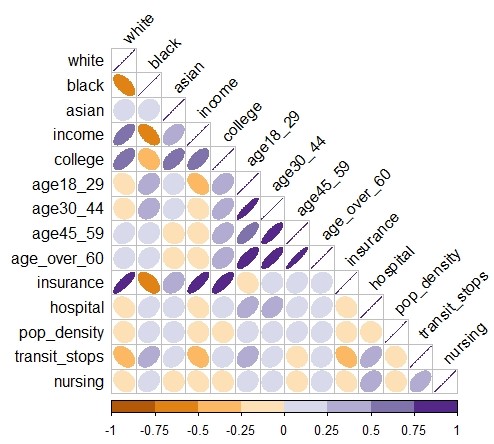}
    \caption{Correlation between transformed covariates. Purple and orange color correspond to positive and negative correlation, respectively. The shade and shape reflect the magnitude of correlation.}
    \label{fig:corr}
\end{figure}

% Not surprisingly, race and median income are highly correlated (see Figure~\ref{fig:corr}) and medical insurance is correlated with multiple socio-demographic variables. On the one hand, this high correlation explains the importance of medical insurance as a predictor in the \texttt{glm-funk} model, potentially as a surrogate for other important variables, and, on the other, serves as a reminder for not over interpreting the results. 

% Furthermore, Figure~\ref{fig:corr} shows the collinearity between some covariates, e.g. medical insurance and white population as well as income and college education.
% The distribution of age groups are also correlated. 
% This partly explains the different effect sizes (and variables identified as ``important'' based on CI/CrI) from the BYM2 model comparing to the others, since lasso-type models are less robust to correlated covariates. 

\begin{table}
\caption{\label{tab:debiased}
Rate ratios and 95$\%$ confidence/credible intervals (CI/CrI) for each model. \newline Rate ratios correspond to standardized and transformed covariates.}

%\begin{tabular}{|l|l|l|l|l|l|l|l|l|}
\scriptsize
\begin{tabular}{l|llll}
\hline
\multirow{2}{*}{}                      & {Full} & {\begin{tabular}[]{@{}c@{}} Independent,\\one intercept\end{tabular}} & {\begin{tabular}[]{@{}c@{}}
Independent,\\individual intercepts\end{tabular}} & {BYM2} \\ %\cline{2-9} 
& RR (95\% CI)      & RR (95\% CI)                 & RR (95\% CI)    & RR (95\% CrI)     \\ \hline
Race                                   & \multicolumn{4}{l}{}   \\ %\hline
\multicolumn{1}{r|}{White}   
& 0.99 (0.95, 1.04) & 1.00 (0.96, 1.05) & 1.00 (0.96, 1.04) & 0.73 (0.24, 2.23) \\ 
\multicolumn{1}{r|}{Black}            
& 1.00 (0.94, 1.06) & 0.67 (0.63, 0.72) & 1.00 (0.94, 1.06) & 0.62 (0.27, 1.38) \\ 
\multicolumn{1}{r|}{Asian}            
&1.00 (0.94, 1.06) & 1.42 (1.33, 1.52) & 1.00 (0.94, 1.06) & 1.71 (0.72, 4.13) \\ [2mm] %\hline
Income                                 
& 0.99 (0.95, 1.04) & 0.90 (0.86, 0.94) & 1.00 (0.95, 1.04) & 0.46 (0.10, 2.14) \\ [2mm] %\hline
College education                      
& 1.00 (0.96, 1.04) & 0.63 (0.60, 0.65) & 1.00 (0.96, 1.04) & 0.85 (0.18, 4.00) \\ [2mm] %\hline
Age                                    
& \multicolumn{4}{l}{}                               \\ %\hline
\multicolumn{1}{r|}{18-29}            
& 1.01 (0.96, 1.05) & 0.27 (0.26, 0.28) & 1.00 (0.96, 1.05) & 0.56 (0.11, 2.90) \\ 
\multicolumn{1}{r|}{30-44}            
& 1.01 (0.97, 1.06) & 4.03 (3.85, 4.21) & 1.00 (0.96, 1.05) & 0.74 (0.10, 5.64) \\ 
\multicolumn{1}{r|}{45-59}            
&1.01 (0.96, 1.05) & 0.53 (0.51, 0.56) & 1.00 (0.96, 1.05) & 0.95 (0.10, 7.92) \\ 
\multicolumn{1}{r|}{\textgreater{}60} 
&  1.13 (1.08, 1.19) & 2.18 (2.08, 2.28) & 1.01 (0.97, 1.06) & 3.14 (0.58, 17.73) \\ [2mm] %\hline
Medical insurance                      
& 1.00 (0.95, 1.04) & 0.83 (0.80, 0.87) & 1.00 (0.96, 1.04) & 1.40 (0.34, 5.69) \\ [2mm] %\hline
Hospitals                              
& 1.01 (0.94, 1.08) & 1.12 (1.05, 1.19) & 1.00 (0.93, 1.07) & 1.00 (0.63, 1.60) \\ [2mm] %\hline
Population density                     
& 0.92 (0.76, 1.10) & 0.59 (0.41, 0.86) & 0.95 (0.78, 1.15) & 0.76 (0.42, 1.21) \\ [2mm] %\hline
Transit stops                          
& 1.02 (0.96, 1.09) & 0.98 (0.92, 1.04) & 1.00 (0.94, 1.06) & 0.91 (0.49, 1.71) \\ [2mm] %\hline
Nursing homes                          
& 1.28 (1.22, 1.36) & 1.28 (1.20, 1.36) & 1.04 (0.97, 1.10) & 1.37 (0.88, 2.13) \\ [2mm] \hline
\end{tabular}
\end{table}

\subsection{COVID-19 model with feature network penalty}
\label{sec:covid2}

In our first analysis in Section~\ref{sec:covid1}, we identified a subset of covariates to include in the model based on the current knowledge of COVID-19 pandemic. This is clearly the right approach if such knowledge exists. However, even in well-studied cases, such as COVID-19, we may not be able to fully identify the relevant covariates that need to be included in the analysis. 

To illustrate the full capabilities of our doubly-regularized approach, in this section we take such an approach and include additional features in the model.
More specifically, in addition to the covariates considered in Section~\ref{sec:covid1}, our expanded analysis includes three types of additional features, namely, environmental variables (proportion of medium and high basins, erosion hazard and landslide hazard), public facilities (number of schools, solid waste facilities and veterans, seniors and human services levy) and public safety facilities (number of fire stations and police stations).
These covariates are less related to COVID-19 deaths based on our current knowledge, and so we investigate whether our model estimates are robust against the inclusion of potentially weak predictors.

The proposed \texttt{glm-funk} framework with both unit and feature network penalties allows us to encourage similarity among coefficients associated with covariates in each of the three new groups of features (environmental, public facility and public safety) through a feature network $G_p$ in which features in the same category are connected.
We report results for both the $\ell_1$ and the $\ell_2$ fusion penalty.
Table~\ref{tab:debiased2} summarizes the estimates, confidence intervals and p-values for these models. 
Estimates and credible intervals from the BYM2 model are presented as a comparison. 
Figure~\ref{fig:rslts_high_dim} visualizes the residuals from each model, and reports the corresponding average negative log-likelihood and RMSE on validation sets across the 10-fold cross-validation. 
Since predictions from the BYM2 model have large variability across the 10 cross-validation folds as in Section~\ref{sec:covid1}, we report its median negative log-likelihood and median RMSE instead.

\begin{figure}
    \centering
    \includegraphics[width=16cm]{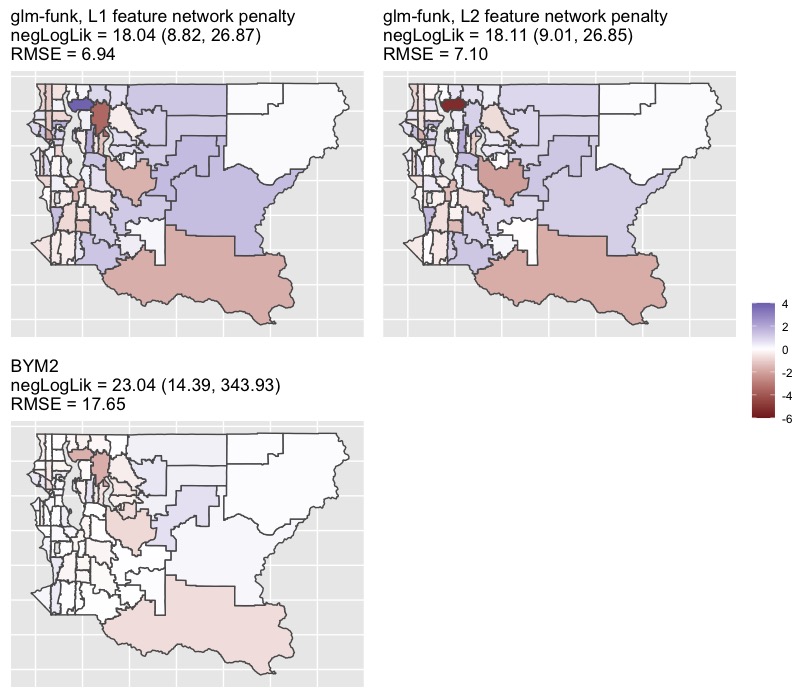}
    \caption{Residuals from each model with added covariates, with cross-validated negative log-likelihood (along with 95\% CI across 10 cross-validation folds) and prediction RMSE reported in the titles. We winsorized the negative log-likelihood for BYM2, replacing the minimum (and maximum) values with the second smallest (and largest) values to mitigate the impact of outliers.}
    \label{fig:rslts_high_dim}
\end{figure}

\begin{table}
\caption{\label{tab:debiased2} Debiased estimates and 95\% confidence/credible intervals \newline from the proposed model with feature network penalty and from the BYM2 model}

% \resizebox{\textwidth}{!}{
\scriptsize
\begin{tabular}{r|lllll}
\hline
\multicolumn{1}{l|}{}                   & \multicolumn{2}{l}{glm-funk, $\ell_1$}                         & \multicolumn{2}{l}{glm-funk, $\ell_2$}                         & BYM2              \\
\multicolumn{1}{l|}{}                   & RR (95\% CI)      & p-value                              & RR (95\% CI)      & p-value                              & RR (95\% CrI)     \\ \hline
\multicolumn{1}{l|}{Race}               &                   &                                      &                   &                                      &                   \\
White                                   & 
0.91 (0.87, 0.95) & \textless{}0.001 & 0.90 (0.86, 0.94) & \textless{}0.001 & 0.74 (0.18, 3.00) \\ 
Black                                   & 
1.09 (1.01, 1.17) & 0.019 & 1.07 (0.99, 1.15) & 0.076 & 0.63 (0.25, 1.55) \\ 
Asian                                   & 
1.00 (0.94, 1.06) & 0.950 & 1.00 (0.94, 1.06) & 0.963 & 1.61 (0.55, 4.82) \\ [2mm]
\multicolumn{1}{l|}{Income}             & 0.96 (0.92, 1.01) & 0.090 & 0.97 (0.92, 1.01) & 0.129 & 0.42 (0.07, 2.71) \\[2mm]
\multicolumn{1}{l|}{College education}  & 1.00 (0.95, 1.05) & 0.972 & 0.99 (0.94, 1.04) & 0.609 & 1.00 (0.09, 10.52) \\[2mm]
\multicolumn{1}{l|}{Age}                &                   &                                      &                   &                                      &                   \\
18-29                                   & 1.02 (0.97, 1.06) & 0.485 & 1.01 (0.97, 1.06) & 0.553 & 0.59 (0.09, 3.89) \\ 
30-44                                   & 1.01 (0.97, 1.06) & 0.524 & 1.03 (0.99, 1.08) & 0.173 & 0.63 (0.06, 6.36) \\ 
45-59                                   & 1.01 (0.96, 1.05) & 0.803 & 1.01 (0.96, 1.05) & 0.815 & 1.17 (0.07, 18.03) \\
\textgreater{}60                        & 1.01 (0.97, 1.06) & 0.596 & 1.02 (0.98, 1.07) & 0.374 & 2.81 (0.32, 25.33) \\ [2mm]
\multicolumn{1}{l|}{Medical insurance}  &  0.92 (0.88, 0.96) & \textless{}0.001 & 0.91 (0.87, 0.95) & \textless{}0.001 & 1.43 (0.25, 8.00) \\ [2mm]
\multicolumn{1}{l|}{Hospitals}          & 0.99 (0.93, 1.07) & 0.873 & 1.01 (0.94, 1.09) & 0.782 & 0.98 (0.56, 1.74) \\ [2mm]
\multicolumn{1}{l|}{Population density} & 0.92 (0.78, 1.09) & 0.352 & 0.88 (0.71, 1.09) & 0.242 & 0.75 (0.38, 1.31) \\ [2mm]
\multicolumn{1}{l|}{Transit stops}      &  1.04 (0.96, 1.13) & 0.341 & 1.05 (0.96, 1.14) & 0.275 & 0.91 (0.34, 2.40) \\ [2mm]
\multicolumn{1}{l|}{Nursing homes}      &  1.37 (1.28, 1.48) & \textless{}0.001 & 1.37 (1.27, 1.47) & \textless{}0.001 & 1.41 (0.80, 2.50) \\ [2mm]
\multicolumn{1}{l|}{Environmental}      &                   &                                      &                   &                                      &                   \\
Basin medium                            & 0.98 (0.88, 1.08) & 0.625 & 0.99 (0.90, 1.10) & 0.881 & 1.13 (0.38, 3.36) \\ 
Basin high                              & 0.97 (0.86, 1.09) & 0.564 & 0.94 (0.83, 1.05) & 0.267 & 0.69 (0.19, 2.42) \\ 
Erosion hazard                          & 0.99 (0.94, 1.04) & 0.800 & 1.01 (0.95, 1.07) & 0.864 & 1.06 (0.46, 2.46) \\ 
Landslide hazard                        & 1.01 (0.93, 1.11) & 0.764 & 0.99 (0.90, 1.08) & 0.809 & 1.18 (0.54, 2.58) \\ [2mm]
\multicolumn{1}{l|}{Public facility}    &                   &                                      &                   &                                      &                   \\
School                                  & 1.00 (0.92, 1.09) & 0.993 & 1.00 (0.92, 1.09) & 0.962 & 1.00 (0.50, 2.04) \\ 
Solid waste                             & 0.99 (0.91, 1.08) & 0.872 & 0.99 (0.91, 1.08) & 0.878 & 1.02 (0.60, 1.75) \\ 
Veteran                                 & 0.98 (0.88, 1.08) & 0.640 & 0.98 (0.89, 1.09) & 0.750 & 0.93 (0.47, 1.83) \\ [2mm]
\multicolumn{1}{l|}{Public safety}      &                   &                                      &                   &                                      &                   \\
Fire                                    & 1.00 (0.92, 1.08) & 0.949 & 1.02 (0.94, 1.10) & 0.647 & 1.08 (0.56, 2.12) \\ 
Police                                  &  1.01 (0.92, 1.11) & 0.803 & 1.01 (0.93, 1.09) & 0.877 & 1.07 (0.56, 2.05) \\ 
\hline
\end{tabular}
% }
\end{table}

As we may expect, the added environmental, public facility and public safety variables are not identified as important predictors of COVID-19 death rates. 
The proposed \texttt{glm-funk} model with both $\ell_1$ and $\ell_2$ feature network penalty outperformed BYM2 model in terms of cross-validated negative log-likelihood as well as RMSE. 
With the inclusion of additional features, the BYM2 model still identifies no important predictor for COVID-19 deaths. 
For the \texttt{glm-funk} models, the statistical significance for the majority of predictors remains unchanged, though the proportion of senior population is no longer a significant covariate for the \texttt{glm-funk} models, and medical insurance coverage and race distribution are identified as significant instead.

This indicates that \texttt{glm-funk} is slightly sensitive to the inclusion of spatially structured (and likely less relevant) variables. 
The importance of race matches our knowledge of COVID-19 burden \citep{price2020hospitalization}, and the insignificance of age (for the $>60$ age group) may be due to the presence of other correlated variables in the model; for example, nursing homes and medical insurance as shown in Figure~\ref{fig:corr}. 

%%%%%%%%%%%
\section{Simulation Studies}
\label{p2sec:sims}
%%%%%%%%%%%
\subsection{Simulation settings}
In this section, we evaluate the performance of the proposed \texttt{glm-funk} models using simulated spatial data.
We generated count responses from the model
\begin{align*}
Y \sim \text{Poisson}( \exp( \alpha + X\beta ) ),
\end{align*}
where $\alpha$ denotes a region-specific random intercept, over a discrete spatial lattice based off the geography of King County and surrounding areas.
The spatial lattice graph of these 204 regions (including ZIP codes where COVID-19 data was unavailable for our data analysis in Section \ref{p2sec:covid}) is shown in Figure~\ref{fig:kcfull} in the Supplementary Material.

Using the method of \citet{rue2005gaussian}, we generated the spatial random effects $\alpha_i$ according to an ICAR model as in Equation~\eqref{eqn:CAR} for $i = 1, \dots, n$, with $\tau = 1$. 
For the fixed effects, we set $\beta_1 = \cdots = \beta_{\frac{s}{2}} = \rho$, $\beta_{\frac{s}{2} + 1} = \cdots = \beta_{s} = -\rho$, and the remaining coefficients to 0; we vary the absolute effect size $\rho$ across simulations.
The feature graph $G_p$ is set to have $2p/s$ disconnected components of size $s/2$ each.
Each component has a single hub node which is connected to the remaining non-hub nodes that do not have any other connections.
In generating the features $X$, each hub node feature is generated as $x_{h} \sim N(0,1)$, with connected non-hub features generated as $x_{nh} \sim N(0.35 x_h, 1)$.
This setup is similar to the one used in the simulations in \citet{li2008network}. 
We fit all models on half of the observations, and reserved the remaining $n/2$ observations as a test set to assess the predictive accuracy.

We consider high-dimensional covariate spaces under the data-generating process specified above.
We also consider the setting where the networks $G_n$ and $G_p$ are not fully informative, to assess the robustness of our method.
To assess the different models implemented, we report the power of the hypothesis test to detect $\beta_j \neq 0$ for the $s$ nonzero coefficients, the Type I error rate among the $p - s$ null coefficients, the test set root mean squared error, and the 95\% coverage probability of confidence intervals for $\beta$. 

%\subsection{High-dimensional simulations}
\subsection{Simulation results}

For this simulation study, we set $p = 300$ and $s = 20$. 
For $G_n$, we use the King County spatial lattice directly, resulting in $n = 204$ regions.
We consider \citet{li2019rnc}'s RNC model with a lasso penalty (i.e., \texttt{glm-funk} with $\gamma_p = 0$) in order to demonstrate the performance of a model that only accounts for structure among $G_n$.
Although \citet{li2019rnc} does not include inference for high-dimensional $\beta$ parameters, inference for this model can be obtained as a special case of our method.
Similarly, we consider a lasso-penalized Poisson regression as a model that does not account for network structure in either the regions or features.

Results for this simulation setting are shown in Figure \ref{fig:hd_inform}. 
We observe that the \texttt{glm-funk} models outperform the models that do not incorporate the feature network information in terms of power.
Comparing \texttt{glm-funk} models with $\ell_1$ and $\ell_2$ smoothing, we see that the $\ell_1$ model achieves the highest power, which makes sense since the connected $\beta$ coefficients are exactly equal to each other.
For out-of-sample predictive performance, the \texttt{glm-funk} models perform the best when the signal strength of features is high. 
The $\ell_1$ model in particular shows meaningfully better predictions, and also better coverage probability.

\begin{figure}[t]
    \centering
    \includegraphics[scale=0.75]{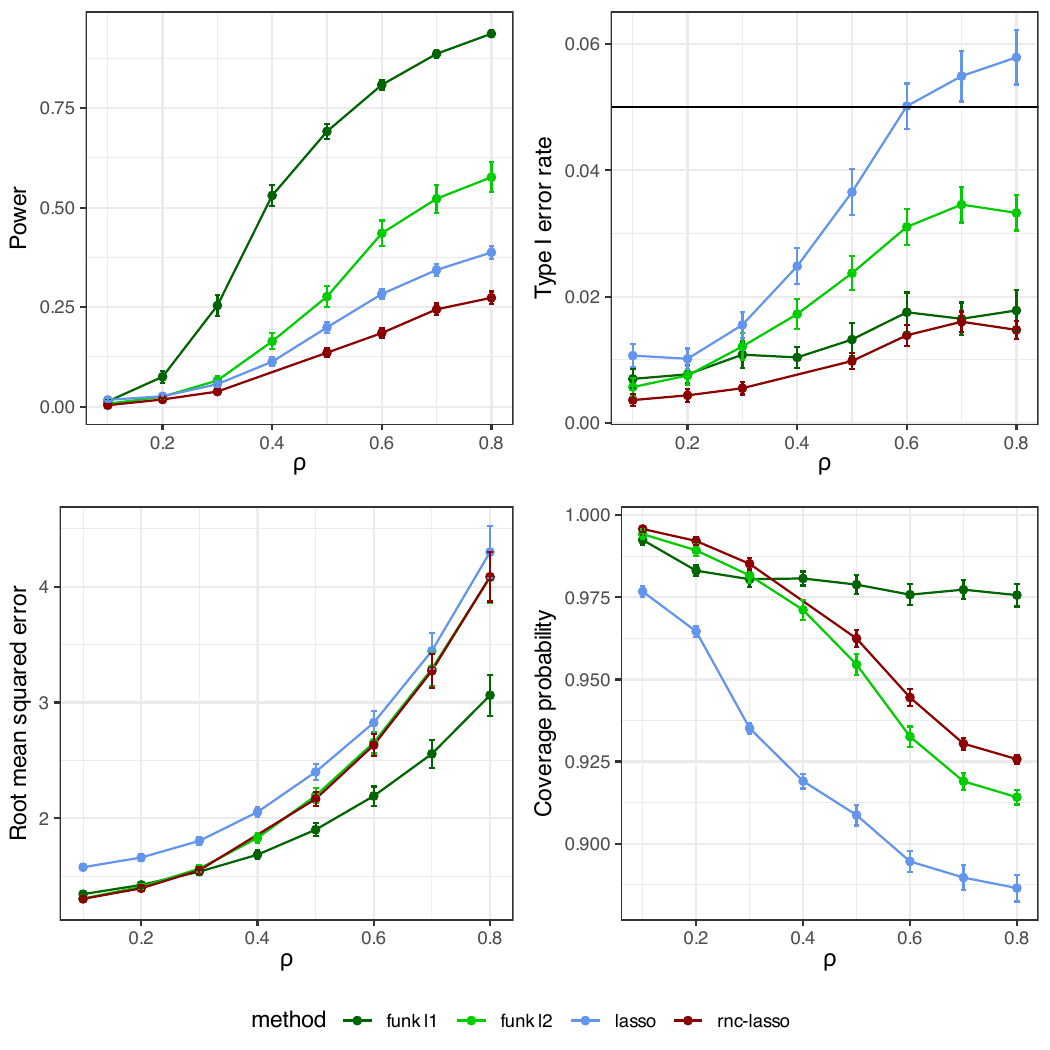}
    \caption{High-dimensional simulation results for fully informative networks. Means over 100 replicates are displayed with standard error bars. Top-left: power, top-right: Type I error rate with nominal 5\% level indicated as horizontal line, bottom-left: test set root mean squared error, bottom-right: 95\% coverage probability of confidence intervals.}
    \label{fig:hd_inform}
\end{figure}

\subsection{Partially informative networks}

We now examine the effect of modeling with partially informative networks in the high-dimensional setting.
For $G_p$, we randomly add intra-component edges to the original network with constant probability 0.002 among all possible edge pairs.
In our simulations, this corresponds to 87 additional edges in $G_p$ (32\% of edges uninformative) on average.
These uninformative edges encourage $\beta$ coefficients from different components to be close to each other.
For $G_n$, we generate the region-specific intercepts independently from a $N(0, 0.24)$ distribution, matching the mean and standard deviation of $\alpha$ in the previous simulation study.
Here, the ICAR model does not apply, and spatial smoothing should not be helpful.

Results for this setting, shown in Figure~\ref{fig:hd_uninform}, indicate that the proposed methods provide robust results when the networks are not fully informative. 
Specifically, the \texttt{glm-funk} models still provide the highest power, though attenuated relative to the fully informative setting.
Their predictive accuracy and coverage probabilities are also reduced, though the $\ell_1$ model still shows an advantage to others.

\begin{figure}[t]
     \centering
     \includegraphics[scale=0.75]{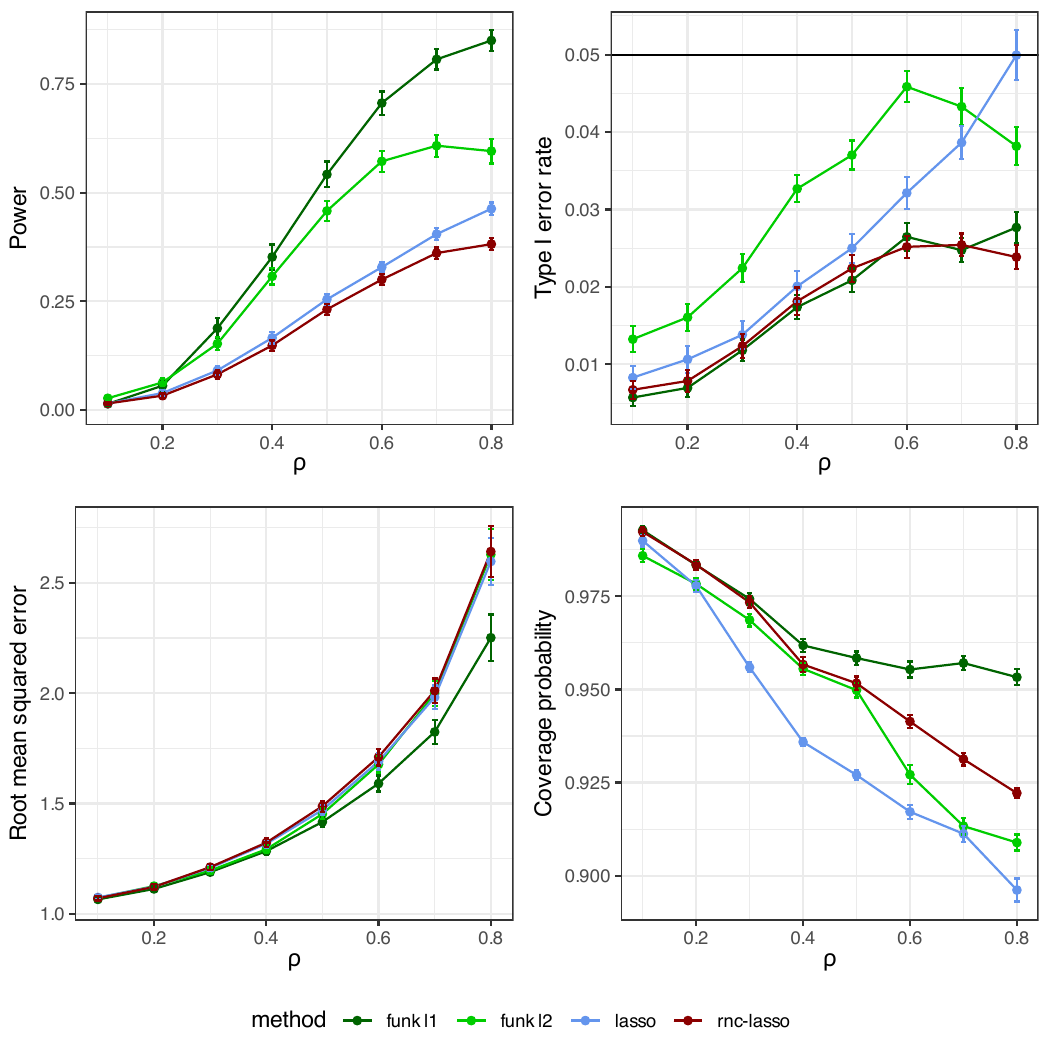}
     \caption{High-dimensional simulation results for partially informative networks. Means over 100 replicates are displayed with standard error bars. Top-left: power, top-right: Type I error rate with nominal 5\% level indicated as horizontal line, bottom-left: test set root mean squared error, bottom-right: 95\% coverage probability of confidence intervals.}
     \label{fig:hd_uninform}
\end{figure}

%%%%%%%%%%%
\section{Discussion}
\label{p2sec:disc}
%%%%%%%%%%%

Motivated by applications in spatial disease modeling, we have developed a new framework for analyzing data with network structure among both observations and features. 
Using doubly regularized generalized linear models, we also obtained valid high-dimensional inference for our model parameters, under potentially informative network structure. 
The proposed methodology---which is implemented in the R package {\tt glmfunk}, available on GitHub---can also be used in other applications involving doubly-structured data. 
Examples include analysis of microbiome data \citep{randolph2018kernel} and multi-view omics data integration \citep{li2018review}.

We applied our methodology to a spatial dataset of COVID-19 death counts, where it demonstrated improved predictive and inferential ability compared to existing methods. 
We showed empirically that spatial smoothing applied to region-specific intercepts results in an appropriate trade-off between region-level bias and attributing outcome variation to covariates.
Our approach also outperformed the spatial Bayesian model that is commonly used for these types of applications, both in out-of-sample prediction and in identifying significant covariates.
Our hypothesis is that \texttt{glm-funk} is better able to adapt to small sample sizes through the tuning of penalty parameters.
Finally, as demonstrated in our simulation studies, \texttt{glm-funk} allows for improved inference in high dimensions while taking advantage of spatial and feature network structure. 

%Similar to approaches based on spatial random effects, our approach is not immune to the challenges of inference for fixed effect parameters due to \emph{spatial confounding} \citep{reich2006effects}. 
%However, our empirical results suggest that our data-adaptive regularization and the additional ridge penalty discussed in Section~\ref{p2sec:funk} lessens the confounding effect. 
%The results from analyzing the well-studied Slovenian cancer data set are case in point: Using a nonspatial GLM with Poisson link, the estimate of the effects of socioeconomic score is $\hat\beta = -0.14$ with 95\% CI: ($-0.18$, $-0.10$). 
%This is very similar to the results reported in \citet{hodges2010adding} who report the posterior mean and 95\% credible intervals $-0.14$ and ($-0.17, -0.10$). \citet{hodges2010adding} report attenuated fixed effect estimate when adding the spatial random effects using the BYM model---posterior mean -$0.02$ and 95\% credible interval ($-0.10,0.06$). 
%Thus, in this case, the addition of spatial random effect amounts to a qualitative change in the importance of the socioeconomic score. 
%In contrast, while the fixed effect parameter is also attenuated in the spatial version of our Poisson GLM---$\hat\beta = -0.09$ and 95\% CI ($-0.15, -0.04$)---the spatial confounding seems to be less severe. As a result, the two models offer qualitatively similar conclusions. 

When the networks are informative for the true data-generating process, we expect our method to show improved prediction and inference compared to standard high-dimensional methods.
When the networks are misspecified or uninformative, our theoretical analysis suggests that the method should still achieve consistent estimation and valid inference for the true regression parameters, if appropriate penalty parameter scaling is used.
This is further borne out in the results of our simulation studies.
However, tuning the network penalty parameters to achieve this in practice may be difficult, as na\"ive cross-validation is not guaranteed to be successful, given the dependencies among the observation units. 
An area for future research would be to incorporate cross-validation for correlated data \citep[similar to, e.g.,][]{cvcorrdata2019} that is compatible with our method.
Although we only considered Laplacian and graph incidence matrices in this paper, other kernels can be easily used within the \texttt{glm-funk} model.
Another possibility would be to avoid adding unit-level intercepts, and instead penalizing the fitted values $X\beta$ directly; that is, using the penalty $\gamma_n \beta' X' L_n X \beta$.

\newpage

\bibliographystyle{plainnat}
\bibliography{manuscriptrefs}

\begin{thebibliography}{48}
\providecommand{\natexlab}[1]{#1}
\providecommand{\url}[1]{\texttt{#1}}
\expandafter\ifx\csname urlstyle\endcsname\relax
  \providecommand{\doi}[1]{doi: #1}\else
  \providecommand{\doi}{doi: \begingroup \urlstyle{rm}\Url}\fi

\bibitem[Aitchison(1982)]{aitchison1982statistical}
John Aitchison.
\newblock The statistical analysis of compositional data.
\newblock \emph{Journal of the Royal Statistical Society: Series B (Statistical Methodology)}, 44\penalty0 (2):\penalty0 139--160, 1982.

\bibitem[Beck and Teboulle(2009)]{beck2009fast}
Amir Beck and Marc Teboulle.
\newblock A fast iterative shrinkage-thresholding algorithm for linear inverse problems.
\newblock \emph{SIAM journal on imaging sciences}, 2\penalty0 (1):\penalty0 183--202, 2009.

\bibitem[Besag(1974)]{besag1974spatial}
Julian Besag.
\newblock Spatial interaction and the statistical analysis of lattice systems.
\newblock \emph{Journal of the Royal Statistical Society: Series B (Statistical Methodology)}, 36\penalty0 (2):\penalty0 192--225, 1974.

\bibitem[Besag and Kooperberg(1995)]{besag1995conditional}
Julian Besag and Charles Kooperberg.
\newblock On conditional and intrinsic autoregressions.
\newblock \emph{Biometrika}, 82\penalty0 (4):\penalty0 733--746, 1995.

\bibitem[Besag et~al.(1991)Besag, York, and Molli{\'e}]{besag1991bayesian}
Julian Besag, Jeremy York, and Annie Molli{\'e}.
\newblock Bayesian image restoration, with two applications in spatial statistics.
\newblock \emph{Annals of the Institute of Statistical Mathematics}, 43\penalty0 (1):\penalty0 1--20, 1991.

\bibitem[B{\"u}hlmann(2013)]{buhlmann2013statistical}
Peter B{\"u}hlmann.
\newblock Statistical significance in high-dimensional linear models.
\newblock \emph{Bernoulli}, 19\penalty0 (4):\penalty0 1212--1242, 2013.

\bibitem[B{\"u}hlmann and van~de Geer(2011)]{buhlmann2011statistics}
Peter B{\"u}hlmann and Sara van~de Geer.
\newblock \emph{Statistics for high-dimensional data: methods, theory and applications}.
\newblock Springer Science \& Business Media, 2011.

\bibitem[Cai et~al.(2019)Cai, Bhattacharjee, Calantone, and Maiti]{cai2019variable}
Liqian Cai, Arnab Bhattacharjee, Roger Calantone, and Taps Maiti.
\newblock Variable selection with spatially autoregressive errors: a generalized moments lasso estimator.
\newblock \emph{Sankhya B}, 81:\penalty0 146--200, 2019.

\bibitem[Chen et~al.(2012)Chen, Lin, Kim, Carbonell, and Xing]{chen2012smoothing}
Xi~Chen, Qihang Lin, Seyoung Kim, Jaime~G Carbonell, and Eric~P Xing.
\newblock Smoothing proximal gradient method for general structured sparse regression.
\newblock \emph{The Annals of Applied Statistics}, 6\penalty0 (2):\penalty0 719--752, 2012.

\bibitem[Chernozhukov et~al.(2021)Chernozhukov, H{\"a}rdle, Huang, and Wang]{chernozhukov2021}
Victor Chernozhukov, Wolfgang~Karl H{\"a}rdle, Chen Huang, and Weining Wang.
\newblock {LASSO-driven inference in time and space}.
\newblock \emph{Annals of Statistics}, 49\penalty0 (3):\penalty0 1702 -- 1735, 2021.
\newblock \doi{10.1214/20-AOS2019}.

\bibitem[Chung(1997)]{chung1997spectral}
Fan~RK Chung.
\newblock Spectral graph theory, regional conference series in math.
\newblock \emph{CBMS, Amer. Math. Soc}, 1997.

\bibitem[Cressie(2015)]{cressie2015spatial}
Noel Cressie.
\newblock \emph{Statistics for spatial data}.
\newblock John Wiley \& Sons, 2015.

\bibitem[Dean et~al.(2001)Dean, Ugarte, and Militino]{dean2001detecting}
CB~Dean, MD~Ugarte, and AF~Militino.
\newblock Detecting interaction between random region and fixed age effects in disease mapping.
\newblock \emph{Biometrics}, 57\penalty0 (1):\penalty0 197--202, 2001.

\bibitem[Dezeure et~al.(2015)Dezeure, B{\"u}hlmann, Meier, and Meinshausen]{dezeure2015high}
Ruben Dezeure, Peter B{\"u}hlmann, Lukas Meier, and Nicolai Meinshausen.
\newblock High-dimensional inference: Confidence intervals, p-values and r-software hdi.
\newblock \emph{Statistical Science}, pages 533--558, 2015.

\bibitem[Duchamp and Stuetzle(2003)]{duchamp2003spline}
Tom Duchamp and Werner Stuetzle.
\newblock Spline smoothing on surfaces.
\newblock \emph{Journal of Computational and Graphical Statistics}, 12\penalty0 (2):\penalty0 354--381, 2003.

\bibitem[Guan and Haran(2018)]{guan2018computationally}
Yawen Guan and Murali Haran.
\newblock A computationally efficient projection-based approach for spatial generalized linear mixed models.
\newblock \emph{Journal of Computational and Graphical Statistics}, 27\penalty0 (4):\penalty0 701--714, 2018.

\bibitem[{Haris} et~al.(2019){Haris}, {Simon}, and {Shojaie}]{gsam2019}
Asad {Haris}, Noah {Simon}, and Ali {Shojaie}.
\newblock {Generalized Sparse Additive Models}.
\newblock \emph{arXiv e-prints}, art. arXiv:1903.04641, Mar 2019.

\bibitem[Holland et~al.(1983)Holland, Laskey, and Leinhardt]{holland1983stochastic}
Paul~W Holland, Kathryn~Blackmond Laskey, and Samuel Leinhardt.
\newblock Stochastic blockmodels: First steps.
\newblock \emph{Social Networks}, 5\penalty0 (2):\penalty0 109--137, 1983.

\bibitem[Hughes(2015)]{hughes2015copcar}
John Hughes.
\newblock copcar: A flexible regression model for areal data.
\newblock \emph{Journal of Computational and Graphical Statistics}, 24\penalty0 (3):\penalty0 733--755, 2015.

\bibitem[Javanmard and Montanari(2013)]{javanmard2013confidence}
Adel Javanmard and Andrea Montanari.
\newblock Confidence intervals and hypothesis testing for high-dimensional statistical models.
\newblock In \emph{Advances in Neural Information Processing Systems}, pages 1187--1195, 2013.

\bibitem[Lee et~al.(2016)Lee, Sun, Sun, Taylor, et~al.]{lee2016exact}
Jason~D Lee, Dennis~L Sun, Yuekai Sun, Jonathan~E Taylor, et~al.
\newblock Exact post-selection inference, with application to the lasso.
\newblock \emph{Annals of Statistics}, 44\penalty0 (3):\penalty0 907--927, 2016.

\bibitem[Leroux et~al.(2000)Leroux, Lei, and Breslow]{leroux2000estimation}
Brian~G Leroux, Xingye Lei, and Norman Breslow.
\newblock Estimation of disease rates in small areas: a new mixed model for spatial dependence.
\newblock In \emph{Statistical models in epidemiology, the environment, and clinical trials}, pages 179--191. Springer, 2000.

\bibitem[Li and Li(2008)]{li2008network}
Caiyan Li and Hongzhe Li.
\newblock Network-constrained regularization and variable selection for analysis of genomic data.
\newblock \emph{Bioinformatics}, 24\penalty0 (9):\penalty0 1175--1182, 2008.

\bibitem[Li et~al.(2019)Li, Levina, and Zhu]{li2019rnc}
Tianxi Li, Elizaveta Levina, and Ji~Zhu.
\newblock Prediction models for network-linked data.
\newblock \emph{The Annals of Applied Statistics}, 13\penalty0 (1):\penalty0 132--164, 2019.

\bibitem[Li et~al.(2018)Li, Wu, and Ngom]{li2018review}
Yifeng Li, Fang-Xiang Wu, and Alioune Ngom.
\newblock A review on machine learning principles for multi-view biological data integration.
\newblock \emph{Briefings in Bioinformatics}, 19\penalty0 (2):\penalty0 325--340, 2018.

\bibitem[Lindgren and Rue(2015)]{lindgren2015bayesian}
Finn Lindgren and H{\aa}vard Rue.
\newblock Bayesian spatial modelling with {R-INLA}.
\newblock \emph{Journal of Statistical Software}, 63\penalty0 (19), 2015.

\bibitem[Morris et~al.(2019)Morris, Wheeler-Martin, Simpson, Mooney, Gelman, and DiMaggio]{morris2019bayesian}
Mitzi Morris, Katherine Wheeler-Martin, Dan Simpson, Stephen~J Mooney, Andrew Gelman, and Charles DiMaggio.
\newblock Bayesian hierarchical spatial models: Implementing the besag york molli{\'e} model in stan.
\newblock \emph{Spatial and Spatio-temporal Epidemiology}, 31:\penalty0 100301, 2019.

\bibitem[Negahban et~al.(2012)Negahban, Ravikumar, Wainwright, and Yu]{negahban2012unified}
Sahand~N Negahban, Pradeep Ravikumar, Martin~J Wainwright, and Bin Yu.
\newblock A unified framework for high-dimensional analysis of $ m $-estimators with decomposable regularizers.
\newblock \emph{Statistical Science}, 27\penalty0 (4):\penalty0 538--557, 2012.

\bibitem[Price-Haywood et~al.(2020)Price-Haywood, Burton, Fort, and Seoane]{price2020hospitalization}
Eboni~G Price-Haywood, Jeffrey Burton, Daniel Fort, and Leonardo Seoane.
\newblock Hospitalization and mortality among black patients and white patients with covid-19.
\newblock \emph{New England Journal of Medicine}, 382\penalty0 (26):\penalty0 2534--2543, 2020.

\bibitem[Rabinowicz and Rosset(2022)]{cvcorrdata2019}
Assaf Rabinowicz and Saharon Rosset.
\newblock Cross-validation for correlated data.
\newblock \emph{Journal of the American Statistical Association}, 117\penalty0 (538):\penalty0 718--731, 2022.

\bibitem[Randolph et~al.(2018)Randolph, Zhao, Copeland, Hullar, and Shojaie]{randolph2018kernel}
Timothy~W Randolph, Sen Zhao, Wade Copeland, Meredith Hullar, and Ali Shojaie.
\newblock Kernel-penalized regression for analysis of microbiome data.
\newblock \emph{The Annals of Applied Statistics}, 12\penalty0 (1):\penalty0 540--566, 2018.

\bibitem[Riebler et~al.(2016)Riebler, S{\o}rbye, Simpson, and Rue]{riebler2016intuitive}
Andrea Riebler, Sigrunn~H S{\o}rbye, Daniel Simpson, and H{\aa}vard Rue.
\newblock An intuitive bayesian spatial model for disease mapping that accounts for scaling.
\newblock \emph{Statistical Methods in Medical Research}, 25\penalty0 (4):\penalty0 1145--1165, 2016.

\bibitem[Ritz and Spiegelman(2004)]{ritz2004equivalence}
John Ritz and Donna Spiegelman.
\newblock Equivalence of conditional and marginal regression models for clustered and longitudinal data.
\newblock \emph{Statistical Methods in Medical Research}, 13\penalty0 (4):\penalty0 309--323, 2004.

\bibitem[Rue and Held(2005)]{rue2005gaussian}
Havard Rue and Leonhard Held.
\newblock \emph{Gaussian Markov random fields: theory and applications}.
\newblock CRC press, 2005.

\bibitem[Rue et~al.(2009)Rue, Martino, and Chopin]{rue2009approximate}
H{\aa}vard Rue, Sara Martino, and Nicolas Chopin.
\newblock Approximate bayesian inference for latent gaussian models by using integrated nested laplace approximations.
\newblock \emph{Journal of the Royal Statistical Society: Series B (Statistical Methodology)}, 71\penalty0 (2):\penalty0 319--392, 2009.

\bibitem[Ruppert et~al.(2003)Ruppert, Wand, and Carroll]{ruppert2003semiparametric}
David Ruppert, Matt~P Wand, and Raymond~J Carroll.
\newblock \emph{Semiparametric regression}.
\newblock Cambridge University Press, 2003.

\bibitem[Sun and Zhang(2012)]{sun2012scaled}
Tingni Sun and Cun-Hui Zhang.
\newblock Scaled sparse linear regression.
\newblock \emph{Biometrika}, 99\penalty0 (4):\penalty0 879--898, 2012.

\bibitem[Tibshirani et~al.(2011)Tibshirani, Taylor, et~al.]{tibshirani2011}
Ryan~J Tibshirani, Jonathan Taylor, et~al.
\newblock The solution path of the generalized lasso.
\newblock \emph{Annals of Statistics}, 39\penalty0 (3):\penalty0 1335--1371, 2011.

\bibitem[van~de Geer(2000)]{van2000empirical}
Sara van~de Geer.
\newblock \emph{Empirical Processes in M-estimation}, volume~6.
\newblock Cambridge University Press, 2000.

\bibitem[van~de Geer et~al.(2014)van~de Geer, B{\"u}hlmann, Ritov, and Dezeure]{van2014asymptotically}
Sara van~de Geer, Peter B{\"u}hlmann, Ya’acov Ritov, and Ruben Dezeure.
\newblock On asymptotically optimal confidence regions and tests for high-dimensional models.
\newblock \emph{Annals of Statistics}, 42\penalty0 (3):\penalty0 1166--1202, 2014.

\bibitem[Vershynin(2018)]{vershynin2018high}
Roman Vershynin.
\newblock \emph{High-dimensional probability: An introduction with applications in data science}, volume~47.
\newblock Cambridge University Press, 2018.

\bibitem[Wahba(1981)]{wahba1981spline}
Grace Wahba.
\newblock Spline interpolation and smoothing on the sphere.
\newblock \emph{SIAM Journal on Scientific and Statistical Computing}, 2\penalty0 (1):\penalty0 5--16, 1981.

\bibitem[Wainwright(2019)]{wainwright2019high}
Martin~J Wainwright.
\newblock \emph{High-dimensional statistics: A non-asymptotic viewpoint}, volume~48.
\newblock Cambridge University Press, 2019.

\bibitem[Wang et~al.(2016)Wang, Sharpnack, Smola, and Tibshirani]{wang2016trend}
Yu-Xiang Wang, James Sharpnack, Alexander~J Smola, and Ryan~J Tibshirani.
\newblock Trend filtering on graphs.
\newblock \emph{The Journal of Machine Learning Research}, 17\penalty0 (1):\penalty0 3651--3691, 2016.

\bibitem[Wasserman and Roeder(2009)]{wasserman2009high}
Larry Wasserman and Kathryn Roeder.
\newblock High dimensional variable selection.
\newblock \emph{Annals of Statistics}, 37\penalty0 (5A):\penalty0 2178, 2009.

\bibitem[Wright(2015)]{wright2015coordinate}
Stephen~J Wright.
\newblock Coordinate descent algorithms.
\newblock \emph{Mathematical Programming}, 151\penalty0 (1):\penalty0 3--34, 2015.

\bibitem[Zhang and Zhang(2014)]{zhang2014confidence}
Cun-Hui Zhang and Stephanie~S Zhang.
\newblock Confidence intervals for low dimensional parameters in high dimensional linear models.
\newblock \emph{Journal of the Royal Statistical Society: Series B (Statistical Methodology)}, 76\penalty0 (1):\penalty0 217--242, 2014.

\bibitem[Zhao and Shojaie(2016)]{zhao2016significance}
Sen Zhao and Ali Shojaie.
\newblock A significance test for graph-constrained estimation.
\newblock \emph{Biometrics}, 72\penalty0 (2):\penalty0 484--493, 2016.

\end{thebibliography}

\newpage

\begin{center}
{\large\bf SUPPLEMENTARY MATERIAL}
\end{center}

\appendix
%%%%%%%%%%%
\section{Technical proofs}
\label{app:proofs}
%%%%%%%%%%%

{\parindent0pt

Here, we prove the results in Section \ref{p2sec:asymp}.
We begin by comparing the target parameters $\theta^*$ to the true parameters $\theta^0$, showing that $\beta^* = \beta^0$ and establishing a bound on $\| \alpha^* - \alpha^0 \|_2$.
We then focus on estimation of the target $\theta^*$ using the $\ell_1$-regularized $\hat{\theta}$.
We derive tail bounds on the empirical process term under both sub-Gaussian and sub-exponential outcomes, which are then used to show $\hat{\theta} \rightarrow \theta^*$ in $\ell_1$ norm. 
Our proof of Theorem \ref{thm:consistency} is similar to that of Theorem 3 in \citet{gsam2019}, which provides fast rates of convergence for generalized sparse additive models. 
A key step in the proof by \citet{gsam2019}, which differentiates it from a similar one in \citet{buhlmann2011statistics}, is handling an intercept term which is not penalized. 
We extend this to our setting, with an $n$-dimensional intercept term that is $\ell_2$ penalized.
We then prove Theorem \ref{thm:inference}, which shows the validity of our de-biased estimator $\hat{b}$ for inference on the true $\beta^0$. 
% end with discussion of l_2 smoothing -- need stronger assumption, get consistency for targets via sparsity, then targets are same as trues

%%%%%%%%%%%
\subsection*{Target vs. true parameters}
%%%%%%%%%%%

We first compare the target parameters $\theta^*$ to the true parameters of interest $\theta^0$. 
We start by noting that the target $\beta^*$ is the same as the true $\beta^0$. 
Using this fact, we characterize the difference in the target and true intercepts; that is, $\| \alpha^* - \alpha^0 \|$.

\begin{lemma}
\label{lemma:targetbetas}
The target parameter $\beta^*$ is equal to the true parameter $\beta^0$. 
\end{lemma}

\begin{proof}
This follows immediately by examining the $\beta$-optimality conditions for both objective functions,
\begin{align*}
0 &= \nabla_\beta \mathbb{E}[\ell(\alpha^* + X\beta^*)] \\
0 &= \nabla_\beta \mathbb{E}[\ell(\alpha^0 + X\beta^0)],
\end{align*}
and by convexity of the loss function.  
\end{proof}

\begin{lemma}
\label{lemma:target}
Under Assumptions \ref{assump:lip} and \ref{assump:scaling}, $\| \alpha^* - \alpha^0 \|_2 = O_p\left( n^c \right)$ for some $c \in (0, \frac{1}{2})$.
\end{lemma}

\begin{proof}

Considering the $\alpha$-optimality condition for the target objective function, we have
\begin{align*}
0 &= \nabla_\alpha \mathbb{E}[\ell(\alpha^* + X\beta^*)] + \gamma_n (L_n + \delta I) \alpha^* \\
 &= \mathbb{E}[\nabla_\alpha \ell(\alpha^* + X\beta^*)] + \gamma_n (L_n + \delta I) \alpha^* \hspace{10pt} \text{(by Assumption \ref{assump:lip})} \\
 &= \mathbb{E}[y - \mu(\alpha^* + X\beta^*)] + \gamma_n (L_n + \delta I) \alpha^*.
\end{align*}

Taking a first order Taylor expansion around the true parameter $\alpha^0$, we obtain
\begin{align}
0 &= \mathbb{E}[y - \mu(\alpha^0 + X\beta^*)] + \gamma_n (L_n + \delta I) \alpha^0 + \left[W_P(\tilde{\alpha} + X\beta^*) + \gamma_n (L_n + \delta I) \right](\alpha^* - \alpha^0),
\label{eq:taylor}
\end{align}
where $\tilde{\alpha}$ is an intermediate point on the line segment between $\alpha^*$ and $\alpha^0$, and $W_P$ is the diagonal matrix of the derivative $\mu'$ over the true data-generating distribution $(y, X) \sim P$ at mean $\mu(\tilde{\alpha} + X\beta^*)$.

Then, using $\beta^* = \beta^0$, $\mathbb{E}[y - \mu(\alpha^0 + X\beta^*)] = 0$, since $\mu(\alpha^0 + X\beta^0)$ is the true conditional mean of $y$. 
Therefore, rearranging \eqref{eq:taylor},
\begin{equation}
    \alpha^0 - \alpha^* = \left[W_P(\tilde{\alpha} + X\beta^*) + \gamma_n (L_n + \delta I) \right]^{-1} \gamma_n (L_n + \delta I) \alpha^0.
\label{eq:alphagap}    
\end{equation}

Taking $\ell_2$ norms in \eqref{eq:alphagap}, we have
\begin{align*}
\| \alpha^0 - \alpha^* \|_2
&\le \gamma_n \| \left[W_P(\tilde{\alpha} + X\beta^*) + \gamma_n (L_n + \delta I) \right]^{-1} \|_2  \| (L_n + \delta I) \alpha^0 \|_2 \\
&\le \gamma_n \lambda_{min}[W_P(\tilde{\alpha} + X\beta^*) + \gamma_n (L_n + \delta I)]^{-1} \| (L_n + \delta I) \alpha^0 \|_2 \\
&\le \gamma_n \lambda_{min}[W_P(\tilde{\alpha} + X\beta^*)]^{-1} \| (L_n + \delta I) \alpha^0 \|_2 \\
&\le U' \gamma_n \| (L_n + \delta I) \alpha^0 \|_2,
\end{align*}
where the final inequality follows from Assumption \ref{assump:lip}.

Then, from part (iii) of Assumption \ref{assump:scaling}, $\| \alpha^0 - \alpha^* \|_2 = O_p\left( n^c \right)$, for $c \in (0, \frac{1}{2})$.

\end{proof}

%%%%%%%%%%%
\subsection*{Control of empirical process}
%%%%%%%%%%%

Recall that our optimization problem can be rewritten as:
$$
\hat\theta = \argmin_{\theta} \left\{\mathbb{P}_n \mathcal{L}(\theta) + \lambda \mathcal{R}(\theta) \right\},
$$
where $\mathcal{L}(\theta) = \ell_i(\alpha_i + x_i'\beta) + \frac{1}{2} \gamma_n \alpha' (L_n + \delta I_n) \alpha$ and $\mathcal{R}(\theta) = \| \beta \|_1 + \frac{\gamma_p}{\lambda} \| J_p \beta \|_1$.

We define the \textit{empirical process term} as 
\begin{align*}
\nu_n(\theta) 
&:= (\mathbb{P}_n - \mathbb{P})\mathcal{L}(\theta) \\
&= (\mathbb{P}_n - \mathbb{P})\ell(\theta),
\end{align*}
since $\mathbb{P}_n (\alpha' (L_n + \delta I_n) \alpha) = \mathbb{P} (\alpha' (L_n + \delta I_n) \alpha)$.

Let $\mathcal{E}(\theta) := \mathbb{P}(\mathcal{L}(\theta^*) - \mathcal{L}(\theta))$, which we define as the \textit{excess risk}. 
Then, similarly as in \citet{gsam2019}, we have the following basic inequality:
\begin{equation}
\mathcal{E}(\hat\theta) + \lambda\mathcal{R}(\hat\theta) \le -[ \nu_n(\hat\theta) - \nu_n(\theta^*)] + \lambda\mathcal{R}(\theta^*),
\label{eq:basicineq}
\end{equation}
which also holds if $\hat\theta$ is replaced by $\tilde\theta := t\hat\theta + (1-t)\theta^*$ for $t \in (0,1)$.

In order to prove that $\hat{\theta} \rightarrow \theta^*$, we need to control the empirical process term. 
We first consider Assumption \ref{assump:tails}, and show the following lemma.

\begin{lemma}
\label{lemma:setT}
Under part (i) of Assumption \ref{assump:tails}, with probability at least $1 - 2 \exp(-n \rho^2 C_1) - C \exp(-n \rho^2 C_2)$, we have, for any $\theta$,
$$
\nu_n(\theta) - \nu_n(\theta^*) \le \rho \left[ \| \alpha - \alpha^* \|_1 + \| \beta - \beta^* \|_1 + \frac{\gamma_p}{\lambda} \| J_p (\beta - \beta^*) \|_1 \right],
$$
where $\rho = O \left( \sqrt{\frac{\log p}{n}} \right)$ and $C, C_1, C_2$ are positive constants independent of $n$ and $p$. 
\end{lemma}

\begin{proof}
This result can be proved almost identically as in \citet{gsam2019}, with the exception of handling the $n-$dimensional intercept $\alpha$. 

Let $x_i \in \mathbb{R}^p$ and $Y_i \in \mathbb{R}$ denote the fixed covariates and response, respectively, for $i = 1, \dots, n$. 
Write the loss for a single observation as 
$$
\ell(\theta) = a Y_i (\alpha_i + x_i'\beta) + h(\alpha_i + x_i'\beta),
$$
for some $a \in \mathbb{R} \setminus \left\{ 0 \right\}$ and function $h: \mathbb{R} \rightarrow \mathbb{R}$. 

Assume $a = 1$, without loss of generality, since this constant will be absorbed into the probability bounds later.
Since $x_i$ are assumed fixed, denoting $\mu_i := \mathbb{E}[Y_i]$, we obtain
\begin{align*}
\nu_n(\theta) 
&= n^{-1} \sum_{i=1}^n (Y_i - \mu_i)(\alpha_i + x_i'\beta).
\end{align*}

Thus, we can write:
\begin{align}
\nu_n(\theta) - \nu_n(\theta^*)
&= n^{-1} \sum_{i=1}^n (Y_i - \mu_i) \left[ (\alpha_i - \alpha_i^*) + \sum_{j=1}^p (\beta_j x_{ij} - \beta_j^* x_{ij}) \right] \nonumber \\
&= n^{-1} \sum_{i=1}^n (Y_i - \mu_i)(\alpha_i - \alpha_i^*) + n^{-1} \sum_{i=1}^n \sum_{j=1}^p (\beta_j x_{ij} - \beta_j^* x_{ij})(Y_i - \mu_i).
\label{eq:ep1}
\end{align}

We now bound the probability that $\nu_n(\theta) - \nu_n(\theta^*)$ exceeds
\begin{equation*}
\rho \left[ \| \alpha - \alpha^* \|_1 + \| \beta - \beta^* \|_1 + \frac{\gamma_p}{\lambda} \| J_p (\beta - \beta^*) \|_1 \right].
\end{equation*}

Consider the following probability involving the first term in \eqref{eq:ep1}: 
$$
P\left( \cfrac{n^{-1} \sum_{i=1}^n (Y_i - \mu_i)(\alpha_i - \alpha_i^*) }{ \| \alpha - \alpha^* \|_1  } \ge \rho \right).
$$
Applying the sub-Gaussian concentration inequality from Lemma 8.2 of \citet{van2000empirical},
\begin{align*}
P\left( \left| \cfrac{n^{-1} \sum_{i=1}^n (Y_i - \mu_i)(\alpha_i - \alpha_i^*) }{ \| \alpha - \alpha^* \|_1  } \right| \ge \rho \right)
&\le 2 \exp \left[- \cfrac{\rho^2}{8(K^2 + \sigma^2_0) \sum_{i=1}^n \gamma_i^2} \right],
\end{align*}
where 
\begin{align*}
\sum_{i=1}^n \gamma_i^2
&= \sum_{i=1}^n \left( \cfrac{\alpha_i - \alpha_i^* }{n \| \alpha - \alpha^* \|_1 } \right)^2 \\
&= \cfrac{1}{n^2} \cfrac{\| \alpha - \alpha^* \|_2^2}{\| \alpha - \alpha^* \|_1^2} \le \cfrac{1}{n^2}.
\end{align*}
Therefore,
\begin{align*}
P\left( \left| \cfrac{n^{-1} \sum_{i=1}^n (Y_i - \mu_i)(\alpha_i - \alpha_i^*) }{ \| \alpha_i - \alpha_i^* \|_1} \right| \ge \rho \right)
&\le 2 \exp \left[- \cfrac{n^2 \rho^2}{8(K^2 + \sigma^2_0)} \right] \\
&= 2 \exp \left(- C_1 n^2 \rho^2 \right),
\end{align*}
where $C_1 = C_1(K, \sigma^2_0)$.

The rest of the proof follows \citet{gsam2019}, bounding the second term in \eqref{eq:ep1}; that is, showing 
$$
\cfrac{1}{n} \sum_{i=1}^n \sum_{j=1}^p (\beta_j x_{ij} - \beta_j^* x_{ij}) (Y_i - \mu_i) \le \rho \left[ \| \beta - \beta^* \|_1 + \frac{\gamma_p}{\lambda} \| J_p (\beta - \beta^*) \|_1 \right]
$$
with high probability.
More specifically, we use a logarithmic entropy bound on the parametric GLM family of functions. 
For each $j$, the following bound on $\delta$-entropy holds with some constant $A_0$ and $T_n = 1$:
\begin{equation}
    \log N(\delta, \mathcal{F}, \| \cdot \|_Q) \le A_0 T_n \log (1/\delta + 1),
\label{eq:logbound}
\end{equation}
where
$$
\mathcal{F} = \left\{ \beta_j x : |\beta_j| + \frac{\gamma_p}{\lambda} | (J_p \beta)_j | \le 1 \right\},
$$
and $Q$ denotes the empirical measure of $x_j$. 
Then, the bound in \eqref{eq:logbound} also holds up to a constant for the class of functions
$$
\mathcal{F'} = \left\{ \cfrac{\beta_j x - \beta_j^* x}{| \beta_j - \beta_j^* | + \frac{\gamma_p}{\lambda} | (J_p (\beta - \beta^*))_j |} : |\beta_j| + \frac{\gamma_p}{\lambda} | (J_p \beta)_j | \le 1 \right\},
$$
for all $j = 1, \dots, p$ \citep{gsam2019}.

Note that since
\begin{align*}
\left| \cfrac{\beta_j x - \beta_j^* x}{| \beta_j - \beta_j^* | + \frac{\gamma_p}{\lambda} | (J_p (\beta - \beta^*))_j |} \right|
&= \cfrac{|\beta_j x - \beta_j^* x|}{| \beta_j - \beta_j^* | + \frac{\gamma_p}{\lambda} | (J_p (\beta - \beta^*))_j |} \\
&\le \cfrac{|\beta_j x - \beta_j^* x|}{| \beta_j - \beta_j^* |} \\
&= |x|,
\end{align*}
functions in $\mathcal{F'}$ are bounded in absolute value by $|x_{ij}| \le R$ (by Assumption \ref{assump:designL2}).
Then, using Dudley's integral bound, that is,
$$
A_0^{1/2} T_n^{1/2} \int_0^R \log^{1/2} \left( \frac{1}{u} + 1 \right) du \le \tilde{A}_0 T_n^{1/2},
$$
by Corollary 8.3 of \citet{van2000empirical}, we have, for all $\delta \ge 2 C \tilde{A}_0 \sqrt{\frac{T_n}{n}}$,
\begin{equation}
P \left( \sup_{\beta_j x \in \mathcal{F}} \left| \cfrac{n^{-1} \sum_{i=1}^n (Y_i - \mu_i)f(\beta_j x_{ij} - \beta_j^* x_{ij})}{| \beta_j - \beta_j^* | + \frac{\gamma_p}{\lambda} | (J_p (\beta - \beta^*))_j |} \right| \ge \delta \right) \le C \exp \left(- \cfrac{n \delta^2}{4 C^2 R} \right).
\label{eq:vdgbound}
\end{equation}

Let $\delta = \rho = \kappa \sqrt{\frac{\log p}{n}}$.
Then, $\rho \ge 2 C \tilde{A}_0 \sqrt{\frac{\log p}{n}} \ge 2 C \tilde{A}_0 \sqrt{\frac{T_n}{n}}$ since $T_n = 1$. 
Thus, applying \eqref{eq:vdgbound} with a union bound yields
\begin{align*}
& P \left(\max_{j = 1, \dots, p} \sup_{\beta_j x \in \mathcal{F}} \left| \cfrac{n^{-1} \sum_{i=1}^n (Y_i - \mu_i)(\beta_j x_{ij} - \beta_j^* x_{ij})}{| \beta_j - \beta_j^* | + \frac{\gamma_p}{\lambda} | (J_p (\beta - \beta^*))_j |} \right| \ge \rho \right) \\
&\le p C \exp \left(- \cfrac{n \rho^2}{4 C^2 R} \right) \\
&= C \exp \left(- \cfrac{n \rho^2}{4 C^2 R} + \log p \right) \\
&= C \exp \left[- n \rho^2 \left( \cfrac{1}{4 C^2 R} - \cfrac{\log p }{n \rho^2} \right) \right].
\end{align*}

Now,
$$
\cfrac{1}{4 C^2 R} - \cfrac{\log p }{n \rho^2} = \cfrac{1}{4 C^2 R} - \cfrac{1}{\kappa^2}
$$
is positive if $\kappa > \max(2 C \sqrt{R}, 2 C \tilde{A}_0)$. 

Thus,
$$
P \left(\max_{j = 1, \dots, p} \sup_{\beta_j x \in \mathcal{F}} \left| \cfrac{n^{-1} \sum_{i=1}^n (Y_i - \mu_i)(\beta_j x_{ij} - \beta_j^* x_{ij})}{| \beta_j - \beta_j^* | + \frac{\gamma_p}{\lambda} | (J_p (\beta - \beta^*))_j |} \right| \ge \rho \right) \\
\le C \exp [-n \rho^2 C_2],
$$
where $C = C(K, \sigma^2_0)$ and $C_2 = C_2(C, R, \kappa) > 0$. 

Therefore, we have
$$
P\left( \left| n^{-1} \sum_{i=1}^n (Y_i - \mu_i)(\alpha_i - \alpha_i^*) \right| \ge \rho \| \alpha - \alpha^* \|_1 \right)
\le 2 \exp \left(- C_1 n^2 \rho^2 \right),
$$
and
$$
P\left( \left| n^{-1} \sum_{i=1}^n \sum_{j=1}^p (\beta_j x_{ij} - \beta_j^* x_{ij}) (Y_i - \mu_i) \right| \ge \rho \left[ \| \beta - \beta^* \|_1 + \frac{\gamma_p}{\lambda} \| J_p (\beta - \beta^*) \|_1 \right] \right)
\le C \exp [-n \rho^2 C_2],
$$
giving bounds for the terms in \eqref{eq:ep1} as claimed.
\end{proof}

The next lemma establishes control of the empirical process term for cases where $Y_i - \mu_i$ is not sub-Gaussian (e.g. for Poisson or exponential GLMs).
In such cases, we require part (ii) of Assumption \ref{assump:tails} instead. 

\begin{lemma}
\label{lemma:setT-subexp}
Under part (ii) of Assumption \ref{assump:tails}, with probability at least $1 - 2 \exp(-n \rho C_3) - 2 \exp(-n \rho^2 C_4)$, the following inequality holds:
$$
\nu_n(\theta) - \nu_n(\theta^*) \le \rho \left[ \| \alpha - \alpha^* \|_1 + \| \beta - \beta^* \|_1 + \frac{\gamma_p}{\lambda} \| J_p (\beta - \beta^*) \|_1 \right],
$$
where $\rho = O \left( \sqrt{\frac{\log p}{n}} \right)$ and $C_3, C_4$ are positive constants independent of $n$ and $p$. 
\end{lemma}

\begin{proof}
As in Lemma \ref{lemma:setT}, we can write
\begin{align}
\nu_n(\theta) - \nu_n(\theta^*)
&= n^{-1} \sum_{i=1}^n (Y_i - \mu_i) \left[ (\alpha_i - \alpha_i^*) + \sum_{j=1}^p (\beta_j x_{ij} - \beta_j^* x_{ij}) \right] \nonumber \\
&= n^{-1} \sum_{i=1}^n (Y_i - \mu_i)(\alpha_i - \alpha_i^*) + n^{-1} \sum_{i=1}^n \sum_{j=1}^p (\beta_j x_{ij} - \beta_j^* x_{ij}) (Y_i - \mu_i).
\label{eq:ep2}
\end{align}

We now want to bound the probability that $\nu_n(\theta) - \nu_n(\theta^*)$ exceeds
$$
\rho \left[ \| \alpha - \alpha^* \|_1 + \| \beta - \beta^* \|_1 + \frac{\gamma_p}{\lambda} \| J_p (\beta - \beta^*) \|_1 \right].
$$

Consider the following probability involving the first term in \eqref{eq:ep2}, 
$$
P\left( \left| \cfrac{n^{-1} \sum_{i=1}^n (Y_i - \mu_i)(\alpha_i - \alpha_i^*) }{ \| \alpha - \alpha^* \|_1  } \right| \ge \rho \right).
$$

By Bernstein's inequality for sub-exponential random variables (Theorem 2.8.2 in \citet{vershynin2018high}), we have
\begin{align*}
P \left( \left| \sum_{i=1}^n \gamma_i (Y_i - \mu_i) \right| \ge \rho \right)
&\le 2 \exp \left[ -c \min \left( \cfrac{\rho^2}{K_{\psi_1}^2 \| \gamma \|_2^2}, \cfrac{\rho}{K_{\psi_1} \| \gamma \|_\infty } \right) \right],
\end{align*}
for some constant $c$, where $K_{\psi_1} = \max_{i \in 1:n} \| Y_i - \mu_i \|_{\psi_1}$ and
$$
\gamma_i = \cfrac{\alpha_i - \alpha_i^*}{n \| \alpha - \alpha^* \|_1}.
$$
Noting that $\| \gamma \|_2^2 \le n^{-2}$ and $\| \gamma \|_\infty \le n^{-1}$,
\begin{align*}
P\left( \cfrac{n^{-1} \sum_{i=1}^n a (Y_i - \mu_i)(\alpha_i - \alpha_i^*) }{ \| \alpha - \alpha^* \|_1  } \ge \rho \right)
&\le 2 \exp \left[ -c \min \left( \cfrac{\rho^2 n^2}{K_{\psi_1}^2}, \cfrac{\rho n}{K_{\psi_1} } \right) \right] \\
&\le 2 \exp (- n \rho C_3),
\end{align*}
where $C_3 > 0$.

Now, considering the second term in \eqref{eq:ep2}, we want to bound, for a specific $j \in \left\{ 1 \dots, p \right\}$, the probability
\begin{equation}
P \left( \left| \cfrac{n^{-1} \sum_{i=1}^n (Y_i - \mu_i)(\beta_j x_{ij} - \beta_j^* x_{ij})}{| \beta_j - \beta_j^* | + \frac{\gamma_p}{\lambda} | (J_p (\beta - \beta^*))_j |} \right| \ge \rho \right).
\end{equation}
Applying Bernstein's inequality again, with
$$
\gamma_i = \cfrac{ \beta_j x_{ij} - \beta_j^* x_{ij} }{n \left( | \beta_j - \beta_j^* | + \frac{\gamma_p}{\lambda} | (J_p (\beta - \beta^*))_j | \right)},
$$
we can see that 
$$
\| \gamma \|_2^2 \le \cfrac{\sum_{i=1}^n X_{ij}^2}{n^2} \le \cfrac{R^2}{n},
$$ 
and 
$$
\| \gamma \|_\infty \le \cfrac{\| X_j \|_\infty}{n} \le \cfrac{R}{n}.
$$
Thus,
\begin{align}
P \left( \left| \cfrac{n^{-1} \sum_{i=1}^n (Y_i - \mu_i)(\beta_j x_{ij} - \beta_j^* x_{ij})}{| \beta_j - \beta_j^* | + \frac{\gamma_p}{\lambda} | (J_p (\beta - \beta^*))_j |} \right| \ge \rho \right)
&\le 2 \exp \left[ -c \min \left( \cfrac{\rho^2 n}{K_{\psi_1}^2  R^2}, \cfrac{\rho n}{K_{\psi_1} R} \right) \right].
\label{eq:ep3}
\end{align}

Now, let $\rho = \kappa \sqrt{\frac{\log p}{n}}$, where
$$
K_{\psi_1} R \le \kappa \le \cfrac{\sqrt{n} K_{\psi_1} R}{\sqrt{\log p}}.
$$
Then, $\rho \le K_{\psi_1} R$, which implies $\cfrac{\rho}{K_{\psi_1} R} \le 1$, or $\cfrac{\rho^2 n}{K_{\psi_1}^2  R^2} \le \cfrac{\rho n}{K_{\psi_1} R}$.

Applying a union bound to \eqref{eq:ep3}, we have
\begin{align*}
& P \left( \max_{j \in 1:p} \left| \cfrac{n^{-1} \sum_{i=1}^n (Y_i - \mu_i)(\beta_j x_{ij} - \beta_j^* x_{ij})}{| \beta_j - \beta_j^* | + \frac{\gamma_p}{\lambda} | (J_p (\beta - \beta^*))_j |} \right| \ge \rho \right) \\
&\le p 2 \exp \left( -c  \cfrac{\rho^2 n}{K_{\psi_1}^2  R^2}  \right) \\
&\le 2 \exp \left( -c  \cfrac{\rho^2 n}{K_{\psi_1}^2  R^2} + \log p \right) \\
&\le 2 \exp \left[ -c n \rho^2 \left( \cfrac{1}{K_{\psi_1}^2  R^2} - \cfrac{\log p}{n \rho^2} \right) \right] \\
&\le 2 \exp \left[ -c n \rho^2 \left( \cfrac{1}{K_{\psi_1}^2  R^2} - \cfrac{1}{\kappa^2} \right) \right] \\
&\le 2 \exp \left( -C_4 n \rho^2 \right),
\end{align*}
where $C_4 >0$ and the last inequality follows since $\frac{1}{K_{\psi_1}^2  R^2} \ge \frac{1}{\kappa^2}$ by the lower bound on $\kappa$. 

\end{proof}

%%%%%%%%%%%
\subsection*{Margin condition} 
%%%%%%%%%%%

Before proving Theorem \ref{thm:consistency}, we show that Assumption \ref{assump:rsc} implies a margin condition \citep{buhlmann2011statistics} for the loss function $\mathcal{L}(\theta) = \ell(\alpha + X\beta) + \frac{1}{2} \gamma_n \alpha' (L_n + \delta I_n) \alpha$ around the target $\theta^*$.

\begin{lemma}
\label{lemma:margin}
Under Assumption \ref{assump:rsc}, the quadratic margin condition,
$$
\mathcal{E}(\theta) \ge G(\| \theta - \theta^* \|),
$$
holds for all $\theta$ satisfying
$$
\| \alpha - \alpha^* \|_1 + \| \beta - \beta^* \|_1 + \frac{\gamma_p}{\lambda} \| J_p (\beta - \beta^* ) \|_1 \le \cfrac{16s\lambda^2}{\rho \phi^2(S) c} + \cfrac{2\gamma_p \| J_p \beta^* \|_1}{\rho},
$$
where $\lambda \ge 8\rho$.
\end{lemma}

\begin{proof}
By Assumption \ref{assump:rsc}, $\mathbb{P} \ell(\theta)$ satisfies restricted strong convexity for $\theta^*$. 
We now show that 
$$
\mathbb{P}\mathcal{L}(\theta) = \mathbb{P}\ell(\theta) + \cfrac{1}{2} \gamma_n \alpha' (L_n + \delta I_n) \alpha
$$
also satisfies restricted strong convexity for $\theta^*$. 
Rewriting 
$$
\cfrac{1}{2} \gamma_n \alpha' (L_n + \delta I_n) \alpha = \cfrac{1}{2} \theta' \tilde{L} \theta,
$$
where 
$$
\tilde{L} = \left[ \begin{array}{cc} \gamma_n(L_n + \delta I_n) & 0 \\ 0 & 0 \end{array} \right],
$$
we can see that $\nabla^2_\theta \frac{1}{2} \gamma_n \alpha' (L_n + \delta I_n) \alpha = \tilde{L}$ is positive semi-definite. 

Therefore, $\mathbb{P} \mathcal{L}(\theta)$ satisfies restricted strong convexity for $\theta^*$.
Hence,
$$
\mathbb{P}(\mathcal{L}(\theta) - \mathcal{L}(\theta^*)) \ge \nabla \mathbb{P}\mathcal{L}(\theta^*)'(\theta - \theta^*) + G(\| \theta - \theta^* \|).
$$
Since $\theta^*$ minimizes $\mathbb{P}\mathcal{L}(\theta)$, we have
$$
\mathbb{P}(\mathcal{L}(\theta) - \mathcal{L}(\theta^*)) = \mathcal{E}(\theta) \ge G(\| \theta - \theta^* \|).
$$
\end{proof}

%%%%%%%%%%%
\subsection*{Proof of Theorem \ref{thm:consistency}}
%%%%%%%%%%%

As a result of Lemma \ref{lemma:setT} or \ref{lemma:setT-subexp}, we have, with high probability, that
\begin{align*}
Z_{M^*}
&:= \sup_{\| \alpha - \alpha^* \|_1 + \mathcal{R}(\theta - \theta^*) \le M^*} | \nu_n(\theta) - \nu_n(\theta^*) | \\
&\le \sup_{\| \alpha - \alpha^* \|_1 + \mathcal{R}(\theta - \theta^*) \le M^*} \rho \left[ \| \alpha - \alpha^* \|_1 + \| \beta - \beta^* \|_1 + \frac{\gamma_p}{\lambda} \| J_p (\beta - \beta^*) \|_1 \right] \\
&= \sup_{\| \alpha - \alpha^* \|_1 + \mathcal{R}(\theta - \theta^*) \le M^*} \rho \left[ \| \alpha - \alpha^* \|_1 + \mathcal{R}(\theta - \theta^*) \right] \\
&\le \rho M^*
\end{align*}

Set $t = \cfrac{M^*}{M^* + \| \hat{\alpha} - \alpha^* \|_1 + \mathcal{R}(\hat{\theta} - \theta^*)}$ and take $\tilde\theta := t\hat\theta + (1-t)\theta^*$. 
Then, 
$$
\| \tilde\alpha - \alpha^* \|_1 + \mathcal{R}(\tilde\theta - \theta^*) \le M^*,
$$
by construction.

Note also that
\begin{equation}
\| \hat\alpha - \alpha^* \|_1 + \mathcal{R}(\hat\theta - \theta^*) = \cfrac{1}{t} \left[ \| \tilde\alpha - \alpha^* \|_1 + \mathcal{R}(\tilde\theta - \theta^*) \right],
\label{eq:tildehat}
\end{equation}
which we will use later to bound $\| \hat\alpha - \alpha^* \|_1 + \mathcal{R}(\hat\theta - \theta^*)$.

Now, starting from the basic inequality \eqref{eq:basicineq},
\begin{align*}
\mathcal{E}(\tilde\theta) + \lambda\mathcal{R}(\tilde\theta) 
&\le -[ \nu_n(\tilde\theta) - \nu_n(\theta^*)] + \lambda\mathcal{R}(\theta^*) \\
&\le Z_{M^*} + \lambda\mathcal{R}(\theta^*) \\
&\le \rho M^* + \lambda\mathcal{R}(\theta^*).
\end{align*}

But,
\begin{align*}
\lambda\mathcal{R}(\theta^*) 
&= \lambda \left[ \| \beta^*_S \|_1 + \frac{\gamma_p}{\lambda} \| J_p \beta^* \|_1 \right] \\
&\le \lambda \left[ \| \beta^*_S - \tilde{\beta}_S \|_1 + \| \tilde{\beta}_S \|_1 + \frac{\gamma_p}{\lambda} \| J_p \beta^* \|_1 \right],
\end{align*}

and
\begin{align*}
\lambda\mathcal{R}(\tilde\theta) 
&= \lambda \left[ \| \tilde\beta_S \|_1 +  \| \tilde\beta_{S^c} \|_1 + \frac{\gamma_p}{\lambda} \| J_p \tilde\beta \|_1 \right] \\
&\ge \lambda \left[ \| \tilde\beta_S \|_1 + \| (\tilde\beta - \beta^*)_{S^c} \|_1 + \frac{\gamma_p}{\lambda} \| J_p (\tilde\beta - \beta^*) \|_1 - \frac{\gamma_p}{\lambda} \| J_p \beta^* \|_1 \right]. 
\end{align*}

Therefore,
\begin{align*}
& \mathcal{E}(\tilde\theta) + \lambda \left[ \| \tilde\beta_S \|_1 + \| (\tilde\beta - \beta^*)_{S^c} \|_1 + \frac{\gamma_p}{\lambda} \| J_p (\tilde\beta - \beta^*) \|_1 - \frac{\gamma_p}{\lambda} \| J_p \beta^* \|_1 \right] \\
&\le \lambda \left[ \| \beta^*_S - \tilde{\beta}_S \|_1 + \| \tilde{\beta}_S \|_1 + \frac{\gamma_p}{\lambda} \| J_p \beta^* \|_1 \right] + \rho M^*.
\end{align*}

Rearranging yields:
\begin{align*}
\mathcal{E}(\tilde\theta) + \lambda \left[ \| (\tilde\beta - \beta^*)_{S^c} \|_1 + \frac{\gamma_p}{\lambda} \| J_p (\tilde\beta - \beta^*) \|_1 \right] \le 2\lambda \| (\tilde\beta - \beta^*)_S \|_1 + 2 \gamma_p \| J_p \beta^* \|_1 + \rho M^*.
\end{align*}

Adding $\lambda \left[ \| \tilde{\alpha} - \alpha^* \|_1 + \|(\tilde\beta - \beta^*)_S \|_1 \right]$ to both sides, we obtain
\begin{equation}
\mathcal{E}(\tilde\theta) + \lambda \left[ \| \tilde{\alpha} - \alpha^* \|_1 + \mathcal{R}(\tilde\theta - \theta^*) \right] \le 2\lambda \| (\tilde\beta - \beta^*)_S \|_1 + \lambda \| \tilde{\alpha} - \alpha^* \|_1 + 2 \gamma_p \| J_p \beta^* \|_1 + \rho M^*.
\label{eq:thm1ineq}
\end{equation}

We have two possible cases for the RHS of \eqref{eq:thm1ineq}.

\textbf{Case I:} $2\lambda \| (\tilde\beta - \beta^*)_S \|_1 + \lambda \| \tilde{\alpha} - \alpha^* \|_1 \le 2\gamma_p \| J_p \beta^* \|_1 + \rho M^*$

\begin{align*}
\mathcal{E}(\tilde\theta) + \lambda \left[ \| \tilde{\alpha} - \alpha^* \|_1 + \mathcal{R}(\tilde\theta - \theta^*) \right]
&\le 4\gamma_p  \| J_p \beta^* \|_1 + 2\rho M^* \\
&\le 2\rho M^* + 2\rho M^* \hspace{10pt} \text{(by the definition of $M^*$)}  \\
&= 4\rho M^* \\
&\le 4\cfrac{\lambda}{8} M^* \\
&= \lambda \cfrac{M^*}{2}
\end{align*}

Therefore, since $\mathcal{E}(\tilde\theta) \ge 0$,
$$
 \| \tilde{\alpha} - \alpha^* \|_1 + \mathcal{R}(\tilde\theta - \theta^*) \le \cfrac{M^*}{2},
$$
and by \eqref{eq:tildehat}, 
\begin{align*}
\| \hat\alpha - \alpha^* \|_1 + \mathcal{R}(\hat\theta - \theta^*) 
&= \cfrac{1}{t} \left[ \| \tilde\alpha - \alpha^* \|_1 + \mathcal{R}(\tilde\theta - \theta^*) \right] \\
&\le \left[ 1 + \cfrac{\| \hat\alpha - \alpha^* \|_1 + \mathcal{R}(\hat\theta - \theta^*)}{M^*} \right] \cfrac{M^*}{2} \\
&= \cfrac{M^*}{2} + \cfrac{\| \hat\alpha - \alpha^* \|_1 + \mathcal{R}(\hat\theta - \theta^*)}{2}.
\end{align*}
Hence,
$$
 \| \hat{\alpha} - \alpha^* \|_1 + \mathcal{R}(\hat\theta - \theta^*) \le M^*.
$$

As a result, we can redo the above arguments replacing $\tilde{\alpha}, \tilde{\beta}$ with $\hat{\alpha}, \hat{\beta}$.

\textbf{Case II:} $2\lambda \| (\tilde\beta - \beta^*)_S \|_1 + \lambda \| \tilde{\alpha} - \alpha^* \|_1 \ge 2\gamma_p  \| J_p \beta^* \|_1 + \rho M^*$

We can bound the RHS of \eqref{eq:thm1ineq} as
\begin{align}
\mathcal{E}(\tilde\theta) + \lambda \left[ \| \tilde{\alpha} - \alpha^* \|_1 + \mathcal{R}(\tilde\beta - \beta^*) \right]
&\le 4\lambda \| (\tilde\beta - \beta^*)_S \|_1 + 2 \lambda \| \tilde{\alpha} - \alpha^* \|_1.
\label{ineq1}
\end{align}

Then, subtracting $\lambda \| (\tilde\beta - \beta^*)_S \|_1 + \lambda \| \tilde{\alpha} - \alpha^* \|_1$ from both sides, we obtain
\begin{align*}
\mathcal{E}(\tilde\theta) + \lambda\|\tilde\beta - \beta^*\|_1 + \gamma_p \| J_p (\tilde\beta - \beta^*) \|_1 - \lambda  \| (\tilde\beta - \beta^*)_S \|_1 + 
&\le 3\lambda \| (\tilde\beta - \beta^*)_S \|_1 + \lambda \| \tilde{\alpha} - \alpha^* \|_1,
\end{align*}
or
\begin{align*}
\mathcal{E}(\tilde\theta) + \lambda\|(\tilde\beta - \beta^*)_{S^c}\|_1 + \gamma_p \| J_p (\tilde\beta - \beta^*) \|_1
&\le 3\lambda \| (\tilde\beta - \beta^*)_S \|_1 + \lambda \| \tilde{\alpha} - \alpha^* \|_1.
\end{align*}

Then, since $\mathcal{E}(\tilde\theta) \ge 0$,
\begin{align*}
\lambda\|(\tilde\beta - \beta^*)_{S^c}\|_1 + \gamma_p \| J_p (\tilde\beta - \beta^*) \|_1
&\le 3\lambda \| (\tilde\beta - \beta^*)_S \|_1 + \lambda \| \tilde{\alpha} - \alpha^* \|_1.
\end{align*}
Dividing by $\lambda$,
\begin{align*}
\|(\tilde\beta - \beta^*)_{S^c}\|_1 + \frac{\gamma_p}{\lambda} \| J_p (\tilde\beta - \beta^*) \|_1
&\le 3 \| (\tilde\beta - \beta^*)_S \|_1 + \| \tilde{\alpha} - \alpha^* \|_1.
\end{align*}

Therefore, the condition in Assumption \ref{assump:comp} is satisfied for $(\tilde{\alpha} - \alpha^*, \tilde\beta - \beta^*)$, so we can use the compatibility condition
$$
\cfrac{\|\tilde{\alpha} - \alpha^*\|_1}{2} + \|(\tilde\beta - \beta^*)_{S}\|_1 \le \cfrac{\| \tilde\theta - \theta^* \| \sqrt{s}}{\phi(s)}.
$$
Plugging this into \eqref{ineq1}, and denoting the convex conjugate of $G$ by $H$, we have:
\begin{align*}
\mathcal{E}(\tilde\theta) + \lambda[ \| \tilde{\alpha} - \alpha^* \|_1 + \mathcal{R}(\tilde\theta - \theta^*) ]
&\le 4\lambda \cfrac{\| \tilde\theta - \theta^* \| \sqrt{s}}{\phi(s)} \\
&\le H\left(\cfrac{4 \lambda \sqrt{s}}{\phi(s)} \right) + G(\| \tilde\theta - \theta^* \|) \\
&\le H\left(\cfrac{4 \lambda \sqrt{s}}{\phi(s)} \right) + \mathcal{E}(\tilde\theta) \\
&= \cfrac{16 s \lambda^2 }{4 c \phi^2(s)} + \mathcal{E}(\tilde\theta) \\
&\le \rho M^* + \mathcal{E}(\tilde\theta) \\
&\le \cfrac{\lambda M^*}{8} + \mathcal{E}(\tilde\theta).
\end{align*}

Finally:
\begin{align*}
 \| \tilde{\alpha} - \alpha^* \|_1 + \mathcal{R}(\tilde\beta - \beta^*) 
&\le  \cfrac{M^*}{8} \le  \cfrac{M^*}{2}.
\end{align*}
Hence, as in Case I, we can show
$$
 \| \hat{\alpha} - \alpha^* \|_1 + \mathcal{R}(\hat\beta - \beta^*) \le M^*.
$$

As a result, we can redo the above arguments replacing $\tilde{\alpha}, \tilde{\beta}$ with $\hat{\alpha}, \hat{\beta}$.

Therefore, in both cases,
\begin{align*}
\mathcal{E}(\hat\beta) + \lambda[ \| \hat{\alpha} - \alpha^* \|_1 + \mathcal{R}(\hat\beta - \beta^*)] 
&\le 4\rho M^* \\
&\le \cfrac{64 s \lambda^2}{c \phi^2(s)} + 8 \gamma_p \| J_p \beta^* \|_1.
\end{align*}

Thus, since $\mathcal{E}(\tilde\theta) \ge 0$,
$$
\| \hat{\alpha} - \alpha^* \|_1 + \| \hat{\beta} - \beta^* \|_1 = O \left( \lambda + \cfrac{\gamma_p}{\lambda} \| J_p \beta^* \|_1 \right),
$$
with probability $1 - 2 \exp(-n^2 \rho^2 C_1) - C \exp(-n \rho^2 C_2)$ by Lemma \ref{lemma:setT} or \ref{lemma:setT-subexp}.

Taking $\lambda = O_p \left( \sqrt{\frac{\log p}{n}} \right)$ and $\gamma_p \| J_p \beta^* \|_1 = o_p(\lambda)$, we have that $\hat{\alpha}$ and $\hat{\theta}$ are $\ell_1$-consistent for $\alpha^*$ and $\beta^*$.

%%%%%%%%%%%
\subsection*{Proof of Theorem \ref{thm:inference}}
%%%%%%%%%%%

Next, we prove Theorem \ref{thm:inference}, which shows the validity of our de-biasing inference procedure.
This proof is very similar to that of Theorem 3.1 in \citet{van2014asymptotically}.
For ease of exposition, we assume the case of GLM families with known scale parameter $\phi = 1$.
However, these results trivially extend to the case with finite known $\phi$, or finite unknown $\phi$ with a consistent estimator $\hat{\phi}$.

Recall that we use the approach of \citet{javanmard2013confidence} for GLMs, defining the de-biased estimator: 
\begin{equation*}
\hat{b} := \hat{\beta} - M \frac{1}{n} \nabla_\beta \ell(\hat{\alpha} + X\hat{\beta}),
\end{equation*}
where $\hat\Theta$ is an estimate of the inverse of $\hat{\Sigma} := \frac{1}{n} \nabla^2_\beta \ell(\hat{\alpha} + X \hat{\beta})$.

Taking a first-order Taylor expansion of the gradient at $\hat{\alpha} + X\hat{\beta}$ around $\hat{\alpha} + X\beta^0$, we have:
\begin{align*}
\nabla_\beta \ell(\hat{\alpha} + X \hat{\beta}) 
&= X'(\mu(\hat{\alpha} + X\beta^0) - y) + X'W(\tilde{q})(X\hat{\beta} - X\beta^0),
\end{align*}
where $W$ is a diagonal matrix of $\frac{\partial \mu}{\partial (\alpha + X_i \beta)}$, and $\tilde{q}$ is an intermediate point in between $\hat{\alpha} + X\hat{\beta}$ and $\hat{\alpha} + X\beta^0$.
Then,
\begin{align*}
\nabla_\beta \ell(\hat{\alpha} + X \hat{\beta}) 
&= X'(\mu(\hat{\alpha} + X\beta^0) - y) + X'W(\hat{\alpha} + X\hat{\beta})X(\hat{\beta} - \beta^0) + X'W(\tilde{q})X(\hat{\beta} - \beta^0) \\
&\hspace{10pt} - X'W(\hat{\alpha} + X\hat{\beta})X(\hat{\beta} - \beta^0).
\end{align*}

Defining $Rem_1 := X'W(\tilde{q})X(\hat{\beta} - \beta^0) - X'W(\hat{\alpha} + X\hat{\beta})X(\hat{\beta} - \beta^0)$, we consider its $\ell_2$ norm:
\begin{align*}
\| Rem_1 \|_2 
&= \| X'W(\tilde{q})X(\hat{\beta} - \beta^0) - X'W(\hat{\alpha} + X\hat{\beta})X(\hat{\beta} - \beta^0) \|_2 \\
&\le \| X'W(\tilde{q}) - X'W(\hat{\alpha} + X\hat{\beta}) \|_2 \| X(\hat{\beta} - \beta^0) \|_2 \\
&\le \| X \|_2 \|W(\tilde{q}) - W(\hat{\alpha} + X\hat{\beta}) \|_2 \| X(\hat{\beta} - \beta^0) \|_2 \\
&\le \| X \|_2 L_{\mu'} \| \tilde{q} - \hat{\alpha} - X\hat{\beta} \|_2 \| X(\hat{\beta} - \beta^0) \|_2 \hspace{10pt} \text{(by Lipschitz condition)} \\
&\le L_{\mu'} \| X \|_2 \| X(\hat{\beta} - \beta^0) \|_2 \| X(\hat{\beta} - \beta^0) \|_2 \hspace{10pt} \text{(by definition of $\tilde{q}$)} \\
&\le L_{\mu'} \| X \|_2^2 \| \hat{\beta} - \beta^0 \|_2^2.
\end{align*}

Then, we have:
\begin{align*}
\hat{b} - \beta^0 
&= (\hat{\beta} - \beta^0) - M \frac{1}{n} \nabla_\beta \ell(\hat{\alpha} + X \hat{\beta}) \\
&= (\hat{\beta} - \beta^0) - M \frac{1}{n} X'(\mu(\hat{\alpha} + X\beta^0) - y) - M \frac{1}{n} X'W(\hat{\alpha} + X\hat{\beta})X(\hat{\beta} - \beta^0) - M \frac{1}{n} Rem_1 \\
&= - \frac{1}{n} M X'(\mu(\hat{\alpha} + X\beta^0) - y) - [M\hat{\Sigma} - I](\hat{\beta} - \beta^0) - M \frac{1}{n} Rem_1 \\
&= \frac{1}{n} M Z_n - Rem_2 - M \frac{1}{n} Rem_1,
\end{align*}
where $Rem_2 := [M\hat{\Sigma} - I](\hat{\beta} - \beta^0)$ and $Z_n := X'(y - \mu(\hat{\alpha} + X\beta^0))$.

Now, by Theorem \ref{thm:consistency} and Lemma \ref{lemma:targetbetas}, we have that $\| M n^{-1} Rem_1 \|_2 = o_p(1)$.
We also have that $\| Rem_2 \|_2 = o_p(1)$ by Theorem \ref{thm:consistency}, Lemma \ref{lemma:targetbetas}, and construction of $M$. 

Also, we can write
\begin{align*}
Z_n &= X' \left( y - \mu(\hat{\alpha} + X\beta^0) + \mu(\alpha^0 + X\beta^0) - \mu(\alpha^0 + X\beta^0) \right) \\
&= X' \left( y - \mu(\alpha^0 + X\beta^0)) + X'(\mu(\hat{\alpha} + X\beta^0)  - \mu(\alpha^0 + X\beta^0) \right) \\
&= X' \left( y - \mu(\alpha^0 + X\beta^0) \right) + Rem_3,
\end{align*}
where $Rem_3 := X' \left( \mu(\hat{\alpha} + X\beta^0)  - \mu(\alpha^0 + X\beta^0) \right)$.

Taking L2 norms,
\begin{align*}
\| Rem_3 \|_2 
&\le \| X \|_2 \| \mu(\hat{\alpha} + X\beta^0)  - \mu(\alpha^0 + X\beta^0) \|_2 \\
&\le \| X \|_2 L_\mu \| \hat{\alpha} - \alpha^0 \|_2 \\
&\le \| X \|_2 L_\mu \left( \| \hat{\alpha} - \alpha^* \|_2 + \| \alpha^* - \alpha^0 \|_2 \right) \\
&\le \| X \|_2 L_\mu \left( \| \hat{\alpha} - \alpha^* \|_1 + \| \alpha^* - \alpha^0 \|_2 \right) \\
&= o_p\left( \cfrac{n^{c_1}}{\sqrt{\log p}} \right),
\end{align*}
where $c_1 < 1$, by Lemma \ref{lemma:target} and Theorem \ref{thm:consistency}.

Hence, $n^{-1} M Rem_3 = o_p(1)$, and
$$
\sqrt{n}(\hat{b} - \beta^0) \longrightarrow_d N \left(0, M \frac{1}{n} E\left[ \nabla \ell(\alpha^0 + X\beta^0) \nabla \ell(\alpha^0 + X\beta^0)' \right] M \right) + o_p(1).
$$

%\clearpage
%%%%%%%%%%%
\section{$\ell_2$ feature network smoothing}
\label{app:l2theory}
%%%%%%%%%%%

In this section, we briefly discuss how the theory given previously can be applied to the \texttt{glm-funk} model with $\ell_2$ feature network smoothing.
Because the generalized ridge penalty $\beta' J_p \beta$ does not represent a norm, it is difficult to work with as part of the regularizer term.
Therefore, we consider it part of the loss function instead, and write the objective function as
$$
\hat\theta = \argmin_{\theta} \left\{\mathbb{P}_n \mathcal{L}_i(\theta) + \lambda \mathcal{R}(\theta) \right\},
$$
where $\mathcal{L}_i(\theta) = \ell_i(\alpha_i + x_i'\beta) + \frac{1}{2} \gamma_n \alpha' (L_n + \delta I_n) \alpha + \frac{1}{2} \gamma_p \beta' L_p \beta $ and $\mathcal{R}(\theta) = \| \beta \|_1$.

Then, since the loss function $\mathbb{P}_n \mathcal{L}_i$ is convex and differentiable in $\theta$, we can apply a similar proof as in Theorem \ref{thm:consistency} to show $\hat{\theta} \rightarrow \theta^*$, where $\theta^* = \argmin_{\theta} \left\{\mathbb{P} \mathcal{L}(\theta) \right\}$.
In order to show that $\| \theta^* - \theta^0 \|$ is negligible, we would need to make a stronger assumption on the target parameters.
That is, if $\gamma_p \|L_p \beta^* \|_2$ and $\gamma_n \| (L_n + \delta I_n) \alpha^* \|_2$ are $o_p(1)$, we can conclude that the target parameter $\theta^*$ asymptotically tends to the true parameter $\theta^0$. 
From here, the validity of our inference procedure given in Theorem~\ref{thm:inference} would follow. 

%\clearpage
%%%%%%%%%%%
\section{Equivalence of RNC and linear mixed models}
\label{app:lmms}
%%%%%%%%%%%

In this section, we briefly discuss the equivalence of the RNC estimator to a linear mixed effects model.
For simplicity, we consider a low-dimensional setting, where $n < p$, and do not incorporate any feature network information. 
We use $\pi$ in this section to denote a density function.

We assume a linear model, 
\begin{align*}
Y &= \alpha + X\beta + \epsilon \\
\alpha &\sim N(0, \phi (L_n + \delta I)^{-1}) \\
\epsilon &\sim N(0, \sigma^2 I),
\end{align*}
where $\alpha \independent \epsilon$.
We first consider \textit{conditional} estimation of $\beta$.
It is easy to see that $Y|\alpha \sim N(\alpha + X\beta, \sigma^2 I)$.
Then, maximizing the log-likelihood $\log \pi(y|\alpha) + \log \pi(\alpha)$ is equivalent to minimizing the objective function
$$
\cfrac{1}{\sigma^2} (y - \alpha - X\beta)'(y - \alpha - X\beta) + \cfrac{1}{\phi} \alpha' (L_n + \delta I) \alpha.
$$
With known $\sigma$ and setting $\phi = \gamma_n^{-1}$, we obtain the RNC objective function
$$
(y - \alpha - X\beta)'(y - \alpha - X\beta) + \gamma_n \alpha' (L_n + \delta I) \alpha.
$$
This relationship holds for other generalized linear models.
Therefore, we can interpret RNC as estimating conditional associations between $y$ and $X$, given correlation induced through random intercepts $\alpha$. 

In the linear model case, we can also use the equivalence of marginal and conditional models due to the additivity of the random effects, and $\alpha$ having mean zero \citep{ritz2004equivalence}.
With known variance, maximizing the likelihood $\pi(y)$ directly is equivalent to minimizing
$$
(y - X\beta)' [\phi (L_n + \delta I)^{-1} + \sigma^2 I]^{-1} (y - X\beta).
$$
Hence, RNC can also be interpreted as marginal estimation of $\beta$ in a mixed model using a generalized least squares estimator.

%%%%%%%%%%%
%\clearpage
\section{Additional simulation studies}
\label{app:moresims}
%%%%%%%%%%%

\begin{figure}[t]
    \centering
    \includegraphics[scale=0.5]{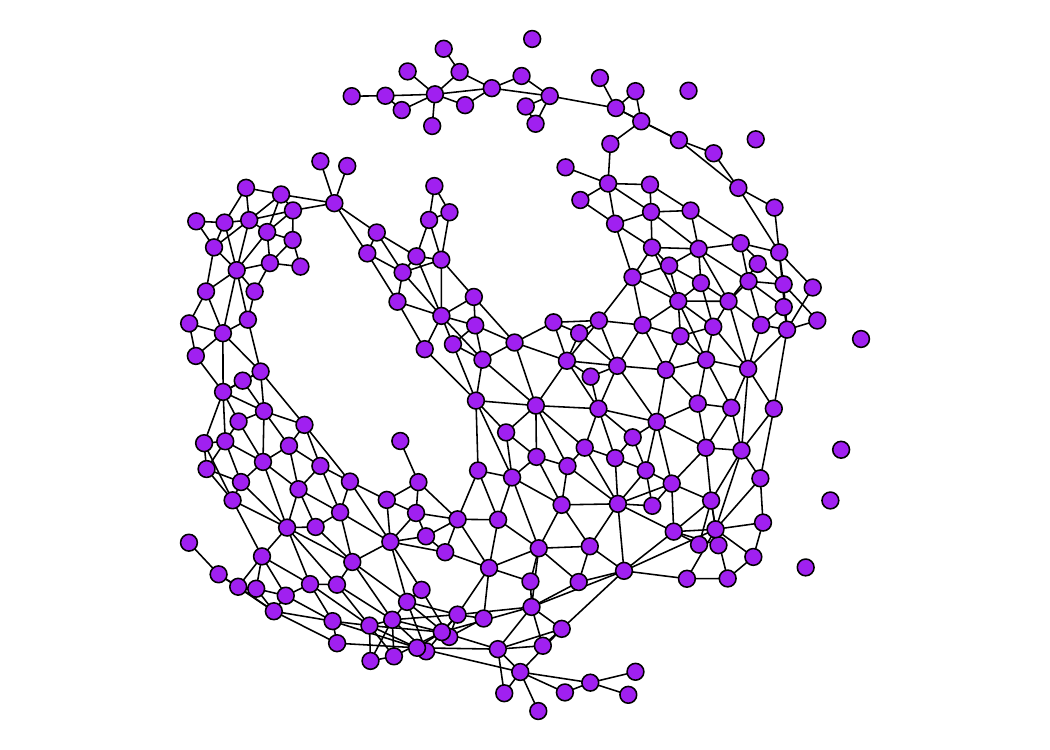}
    \caption{Spatial lattice corresponding to $204$ King County and surrounding area ZIP codes.}
    \label{fig:kcfull}
\end{figure}

In this section, we report the results of additional simulation studies. 
We consider estimation using Grace-type penalization \citep{li2008network}; that is, where we only incorporate feature network information.
We also examine the effect of having less information in the feature network $G_p$, by deleting true edges.

We generate binary data from the following model:
$$
P(Y = 1 | X = x) = \text{expit}(\alpha + X\beta)
$$
with $p = 300, n = 200,$ and $s = 20$. 
The feature and $\beta$ coefficients are set similarly as in the main paper simulations.
However, instead of using a spatial lattice, we set $G_n$ to be a stochastic block model \citep{holland1983stochastic} as in \citet{li2019rnc}.
We divide the observed units into five fully-connected blocks with equal probability of membership.
The unit-level intercepts are then generated from normal distributions that differ between blocks.
Specifically, the block means considered are -4, -2, 0, 2, and 4, so that the intercepts have meaningful effects on $P(Y = 1|X = x)$.
All distributions have a common standard deviation of 0.2.
We select tuning parameters via 5-fold cross-validation; $2n$ observations are generated, with $n$ observations each used for training and as a test set. 
We assume knowledge of the full graph $G_n$ over the $2n$ units, and report the average powers and Type I error rates at the 0.05 significance level, and the test set logistic deviance over 100 simulated datasets.

%%%%%%%%%%%
\subsection*{Grace estimators}
%%%%%%%%%%%

We now compare the \texttt{glm-funk} fits to estimators which do not incorporate any unit network information.
When using $\ell_2$ feature network smoothing, this estimator is the same as the Grace estimator of \citet{li2008network}. 
We also examine the corresponding estimator with $\ell_1$ smoothing. 

Results are shown in Figure \ref{fig:grace1}.
We see that the Grace estimators perform very similarly to the \texttt{glm-funk} estimators in these settings.
Notably, these methods have very similar powers. 
\texttt{glm-funk} shows a slight advantage in terms of prediction.

\begin{figure}[t]
    \centering
    \includegraphics[scale=0.75]{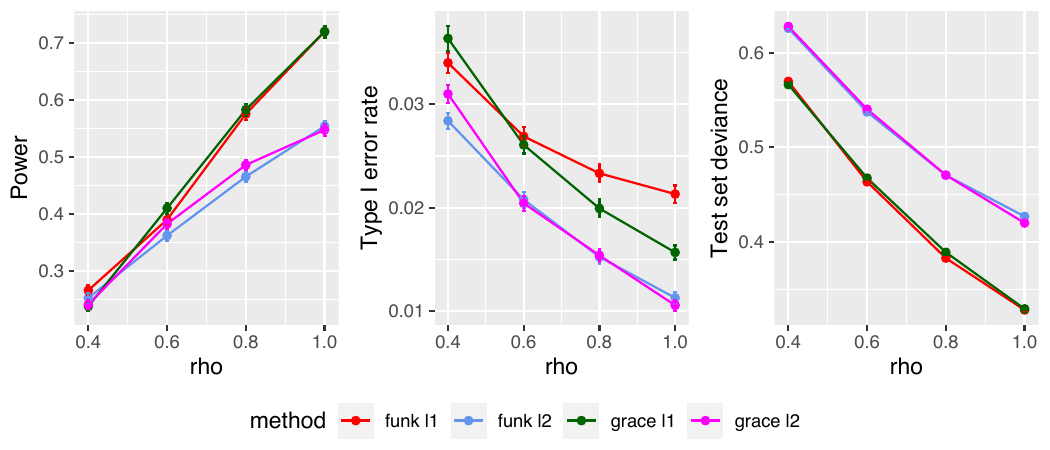}
    \caption{Simulation results for fully informative networks, comparing \texttt{glm-funk} and Grace estimators. Means over 100 replicates are displayed with standard error bars. Left: power, middle: Type I error rate, right: test set deviance.}
    \label{fig:grace1}
\end{figure}

%%%%%%%%%%%
\subsection*{Deleted true edges}
%%%%%%%%%%%

We now consider the same data-generating process, where knowledge of the feature network is restricted to the inactive set of features.
That is, we remove the edges connecting the $s$ active features in $G_p$. 
Note that this network is \textit{less} informative than the full $G_p$, but not \textit{uninformative}, since we do not have any false edges.
We compare the same set of estimators as in the main paper.
As shown in Figure \ref{fig:lessinfo}, the \texttt{glm-funk} estimators perform worse, but still maintain an advantage over the lasso-based methods. 

\begin{figure}[t]
    \centering
    \includegraphics[scale=0.75]{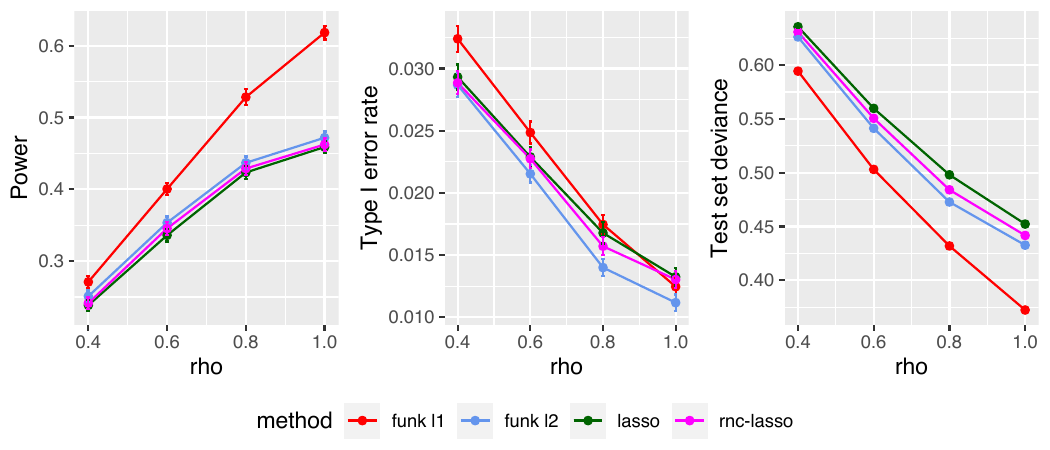}
    \caption{Simulation results for less informative networks. Means over 100 replicates are displayed with standard error bars. Left: power, middle: Type I error rate, right: test set deviance.}
    \label{fig:lessinfo}
\end{figure}

}

\end{document}